\documentclass[12 pt]{article}
\usepackage{bbm}
\usepackage[longnamesfirst]{natbib}
\usepackage{amssymb}
\usepackage{amsmath}
\usepackage{amsthm}
\usepackage{setspace}
\usepackage[colorlinks]{hyperref}
\usepackage{graphicx}
\usepackage{caption}
\usepackage{subcaption}
\usepackage{enumitem}
\setlist{noitemsep,parsep=6pt,partopsep=0pt,topsep=0pt}
\usepackage[margin=1in]{geometry}
\usepackage{verbatim}
\usepackage{sectsty}
\usepackage[svgnames]{xcolor}
\usepackage{subfiles}
\usepackage{xargs}
\usepackage{titlesec}
\usepackage[colorinlistoftodos,textsize=tiny]{todonotes}
\usepackage[capitalize,nameinlink]{cleveref}
\usepackage[normalem]{ulem}
\usepackage[font={small,it},justification=justified]{caption}
\hypersetup{
   pdftitle={Sequential Veto Bargaining with Incomplete Information},
    pdfauthor={S. Nageeb Ali, Navin Kartik, Andreas Kleiner},
    citecolor=DarkBlue,
    bookmarksnumbered=true,
    urlcolor=Indigo,linkcolor=DarkBlue
}
 \theoremstyle{remark}

\theoremstyle{plain}
\newtheorem{theorem}{Theorem}

\newtheorem{assumption}{Assumption}

\newtheorem{definition}{Definition}

\newtheorem{lemma}{Lemma}

\newtheorem{proposition}{Proposition}

\renewcommand{\epsilon}{\varepsilon}
\newcommand{\citepos}[1]{\citeauthor{#1}'s (\citeyear{#1})} 
\newcommand{\R}{\mathbb{R}}
\newcommand{\E}{\mathbb{E}}

\def\Reals{\mathbb R}
\newcommand{\uv}{\underline{v}}
\newcommand{\ov}{\overline{v}}
\newcommand{\oP}{\overline{P}}
\DeclareMathOperator{\Supp}{supp}
\newcommand{\ed}{(\epsilon,\delta)}

\usetikzlibrary{patterns, decorations.pathreplacing, arrows.meta,calc}

\let \savenumberline \numberline
\def \numberline#1{\savenumberline{#1.}}

\makeatletter
  \renewcommand\@seccntformat[1]{\csname the#1\endcsname.{\hskip.7em\relax}} 
\makeatother
\makeatletter
\renewenvironment{proof}[1][\proofname] {\par\pushQED{\qed}\normalfont\topsep6\p@\@plus6\p@\relax\trivlist\item[\hskip\labelsep\bfseries#1\@addpunct{.}]\ignorespaces}{\popQED\endtrivlist\@endpefalse}
\makeatother

\newcommand{\mailto}[1]{\href{mailto:#1}{\texttt{#1}}} 

\let\oldfootnote\footnote
\renewcommand\footnote[1]{\oldfootnote{\hspace{.5mm}#1}}

\titlespacing\section{0pt}{10pt plus 2pt minus 2pt}{4pt plus 2pt minus 2pt} 
\titlespacing\subsection{0pt}{6pt plus 2pt minus 2pt}{2pt plus 2pt minus 2pt} 
\titlespacing\subsubsection{0pt}{6pt plus 2pt minus 2pt}{0pt plus 2pt minus 2pt} 
\titlespacing{\paragraph}{%
  0pt}{
  0.5\baselineskip}{
  1em}

\newcommand{\appendixref}[1]{\hyperref[#1]{Appendix \ref{#1}}}

\usepackage{xcolor}
\usepackage{pgf}
\usepackage{tikz}
\usetikzlibrary{shapes,arrows,automata,fit}

\tikzstyle{info}=[circle,thick,draw=black,fill=black!25,minimum size=4mm]
\tikzstyle{uninfo}=[circle,thick,draw=black,fill=white,minimum size=4mm]
\tikzstyle{inforecog}=[circle,line width=1mm,draw=black!50,fill=black!25,minimum size=4mm]
\tikzstyle{uninforecog}=[circle,line width=1mm,draw=black!50,fill=white,minimum size=4mm]

\tikzstyle{traded}=[draw, line width=1mm]
\tikzstyle{recog}=[draw=black!50, line width=1mm]

\sectionfont{\color{DarkRed}}
\subsectionfont{\color{DarkRed}}
\subsubsectionfont{\color{DarkRed}}
\newcommandx{\nageeb}[2][1=]{\todo[linecolor=blue,backgroundcolor=blue!25,bordercolor=blue,#1]{#2}}
\newcommandx{\andreas}[2][1=]{\todo[linecolor=black,backgroundcolor=black!25,bordercolor=black,#1]{#2}}
\newcommandx{\navin}[2][1=]{\todo[linecolor=red,backgroundcolor=red!25,bordercolor=red,#1]{#2}} 

\begin{document}

\begin{titlepage}

\title{Sequential Veto Bargaining with Incomplete Information\thanks{We thank Laura Doval, Faruk Gul, Marina Halac, Emir Kamenica, Juan Ortner, Daniel Rappoport, Jan Zapal, the Co-editor (Alessandro Lizzeri), four anonymous referees, 
and various audiences for useful discussions. We gratefully acknowledge financial support from NSF Grants SES-2018948 (Kartik) and SES-2018983 (Kleiner). C\'esar Barilla and Yangfan Zhou provided exemplary research assistance. }}
\author{
S. Nageeb Ali\thanks{Department of Economics, Pennsylvania State University. Email: \mailto{nageeb@psu.edu}.}
\and
Navin Kartik\thanks{Department of Economics, Columbia University.  Email: \mailto{nkartik@columbia.edu}.}
\and 
Andreas Kleiner\thanks{Department of Economics, Arizona State University.  Email:  \mailto{andreas.kleiner@asu.edu}.}
}

\maketitle

\begin{abstract}
\noindent 

We study sequential bargaining between a proposer and a veto player. Both have single-peaked preferences, but the proposer is uncertain about the veto player’s ideal point. The proposer cannot commit to future proposals. When players are patient, there can be equilibria with Coasian dynamics: the veto player's private information can largely nullify proposer's bargaining power. Our main result, however, is that under some conditions there also are equilibria in which the proposer obtains the high payoff that he would with commitment power. The driving force is that the veto player's single-peaked preferences give the proposer an option to ``leapfrog'', i.e., to secure agreement from only low-surplus types early on to credibly extract surplus from high types later. Methodologically, we exploit the connection between sequential bargaining and static mechanism design.

\end{abstract}
\thispagestyle{empty} 

\end{titlepage}

\setstretch{1.1}

\clearpage
	\setcounter{tocdepth}{2}
	\tableofcontents
	\thispagestyle{empty}
	\clearpage
\setcounter{page}{1}

\onehalfspacing

\section{Introduction}

\begin{quote}
{``\textit{If the Congress returns the bill having appropriately addressed these concerns, I will sign it.  For now, I must veto the bill.}''\footnote{Closing of \href{https://obamawhitehouse.archives.gov/the-press-office/2016/07/22/veto-message-president-hr-1777}{Obama's Veto Message} when he vetoed H.R. 1777.}} \\ --- President Barack Obama
\end{quote}

An important feature of U.S. politics is that legislatures (e.g., the Congress or a State Assembly) send bills to executives (e.g., the President or a Governor) who can veto them, and conversely, executives must secure confirmation from legislatures for certain appointments (e.g., to the Supreme Court and the Federal Reserve Board). More broadly, there are many contexts in which one party or group makes proposals and another decides whether to approve them. For instance, search committees put forward candidates for approval by their organizations, Boards of Directors may require sign-off from shareholders on certain initiatives, and some public school districts require citizens to ratify the budget proposed by their school boards.

In an influential paper, \citet{RR:78} introduced a framework to study \emph{veto bargaining}, i.e., bargaining over a one-dimensional policy between two players who have single-peaked preferences. Only one player, Proposer, has the power to make proposals; the other player, Vetoer, decides whether to accept a proposal or reject it and preserve the status quo. \citeauthor{RR:78} assumed complete information—specifically, Proposer knows Vetoer’s preferences—and a single take-it-or-leave-it proposal. These are important benchmarks, but for many applications both assumptions ought to be relaxed: Proposer may be uncertain about Vetoer’s preferences, and, as illustrated in our epigraph, Proposer can make sequential proposals.

Sequential veto bargaining with incomplete information presents rich possibilities for learning and signaling. When a proposal is rejected, Proposer updates about Vetoer's preferences and might modify his proposal in response. Anticipating that, Vetoer has an incentive to strategically reject proposals that she prefers over the status quo in order to extract proposals she likes even more. (Consider our epigraph, again.) But then, to what extent does Proposer actually benefit from making multiple proposals?

Existing work on these issues primarily undertakes only a two-period analysis (e.g., \citealp{Cameron:00}, pp.~110-116; \citealp{CM:04}, Section 4).\footnote{We discuss two exceptions, \citet{RR:79} and \citet{CE:94}, in \autoref{sec:discussion}.} But there are limitations to models with a short bargaining horizon. On the one hand, being able to make proposals repeatedly may allow Proposer to reap benefits from screening Vetoer's type. On the other hand, a short horizon confers significant commitment power to Proposer.

The implications of a long horizon have been studied in the neighboring arena of bargaining between a seller and a buyer with privately-known valuation. There, following the classic Coase Conjecture \citep{Coase:72}, it has been
shown that if offers can be made indefinitely and players are patient, then 
lack
of commitment wipes out the seller's bargaining power. The outcome is (approximately) that the buyer only pays her lowest possible valuation so long as it is common knowledge that there are gains from trade.\footnote{This point has been established for the ``gap case'' and, subject to a ``stationary equilibirum'' qualification, also for the ``no gap case'' \citep*{FLT:85,GSW:86}. \cite{AD:89} provide an important counterpoint in the no gap case with non-stationary equilibria.} Applying Coasian logic to veto bargaining would suggest that because sequential rationality compels Proposer to repeatedly moderate future proposals, an inability to commit would significantly hurt Proposer.

Accordingly, the goal of our paper is to study sequential veto bargaining with incomplete information in an infinite-horizon model with patient players. Our main result is that, contrary to a Coasian intuition, the lack of commitment need not harm Proposer. More specifically, we establish that under certain conditions, \emph{if players are patient, Proposer can achieve a payoff that is arbitrarily close to his payoff with commitment power} (\autoref{prop:commitment_achievable_general}).

Central to this result is Proposer's ability to \emph{leapfrog}: he may initially propose a policy that is far from his own interests, targeting acceptance by ``low'' Vetoer types whose ideal points are further away from his and closer to the status quo. Upon rejection, Proposer concludes that Vetoer's ideal point is closer to his own preferred policy. He is then able to extract surplus from these ``high'' types because it is then credible to only offer policies that are even closer to his own ideal point. Put differently, by securing initial acceptance from (only) low types, leapfrogging limits the implications of sequential rationality for subsequent policy moderation, so much so that Proposer is not harmed by the lack of commitment.

Leapfrogging is viable in our model because Vetoer has single-peaked preferences: there are policies that low types are willing to accept and high types are not, given suitable subsequent policy proposals. By contrast, in the canonical model 
of seller-buyer bargaining, all buyer types prefer low to high prices. Offering low prices early on to subsequently charge high-value buyers a higher price would be futile; indeed, any equilibrium in seller-buyer bargaining features decreasing prices with the so-called \emph{skimming} property: the current price is always accepted by  an interval of the highest-value buyer types. 

After presenting our model in \Cref{sec:model}, we use a two-type example in \Cref{sec:twotypes} to develop the logic of leapfrogging. We first show how the option to leapfrog implies that, if an equilibrium exists, there is one that achieves a high Proposer payoff. Our option-based argument is succinct, but leaves open whether and how leapfrogging can be supported in an equilibrium. Accordingly, we also explicitly construct a high Proposer payoff equilibrium that uses leapfrogging (\autoref{prop:twotypes}).

We turn in \autoref{sec:general} to a setting with a continuum of types and Vetoer preferences given by a quadratic loss function. As is familiar in sequential bargaining, an upper bound on Proposer's payoff when he can commit to a strategy in the dynamic game is provided by an auxiliary static mechanism design problem (\autoref{lemma:static_provides_upper_bound}). This static problem has been studied recently by \cite{KKVW:21}; we assume that what they call ``interval delegation'' is an optimal mechanism. \autoref{prop:commitment_achievable_general} then establishes our main result: the static mechanism design payoff can be (approximately) achieved in a sequential veto bargaining equilibrium. Our argument is non-constructive, but crucially exploits Proposer's option to leapfrog in the dynamic game and certain properties of the optimal mechanism (\autoref{lemma:induction_base}). Combining \autoref{lemma:static_provides_upper_bound} and \autoref{prop:commitment_achievable_general}, we conclude that Proposer can achieve (approximately) the same payoff in an equilibrium as he could by committing to a strategy in the dynamic game.

In \autoref{sec:equilibria}, we show that there can be multiple equilibrium outcomes. \autoref{sec:skimmingequilibrium} constructs, under reasonable conditions, a ``skimming equilibrium'' that features Coasian dynamics: Proposer starts with demanding proposals but compromises rapidly, so much so that Vetoer (approximately) gets her ideal point unless it is sufficiently extreme.  In some cases this outcome is a lower bound on Proposer's equilibrium payoff, and an upper bound on Vetoer's. In \autoref{sec:highpayoffeqm}, we build on the skimming equilibrium to explicitly describe the dynamics of a leapfrogging equilibrium that delivers (approximately) Proposer's commitment payoff. Proposer begins by leapfrogging with a low offer, and upon rejection skims among the remaining high types. Although intuitive, this approach bootstraps on the ``bad'' skimming equilibrium by using it as a punishment if Proposer deviates, reminiscent of \citet{AD:89}. By contrast, our non-constructive proof of \autoref{prop:commitment_achievable_general} does not presume existence of a low-payoff equilibrium.
In \Cref{sec:leapfroggingnecessary}, we establish that leapfrogging is sometimes necessary to achieve Proposer's commitment payoff.

As there can be a range of equilibrium payoffs, our analysis calls attention to the role of ``norms''---equilibrium selection---in veto bargaining. 
In particular, if the norm favors Proposer, then the ability to make multiple proposals is always valuable to Proposer; however, under an unfavorable norm, in some environments Proposer could be worse off than if he could only make a single take-it-or-leave-it offer.

\Cref{sec:discussion} relates our work to the existing literature on veto and Coasian bargaining. \Cref{sec:conclusion} concludes.

\section{Model}
\label{sec:model}

Proposer (he) and Vetoer (she) jointly choose a policy or action $a \in \mathbb{R}$. In each period $t=0,1,2,\ldots$, so long as agreement has not already been reached, Proposer makes a proposal $a_t \in \mathbb{R}$ that Vetoer can accept or reject. The game ends when Vetoer accepts a proposal. Both players share a common discount factor $\delta\in [0,1)$. If agreement is reached in some period $T$ on action $a_T$, then 
Proposer's payoff is $\delta^{T}u(a_T)$ and Vetoer's is $\delta^{T} u_V(a_T,v)$; both players' payoffs are $0$ if agreement is never reached. The variable $v\in \Reals$ in Vetoer's payoff is her private information, or type, drawn ex ante from some  cumulative distribution $F$. We interpret the players' payoffs as arising from flow utilities $u$ and $u_V$ when a status-quo policy $0$ is implemented in every period from $0$ to $T-1$ and the agreement policy $a_T$ is implemented forever starting from period $T$, with a normalization that both players' utilities from the status quo is $0$. That is, a player's utility from a policy is his/her gain from that policy relative to the status quo. We assume both players have strictly single-peaked preferences, with Proposer's ideal point being $1$ and Vetoer's $v$. That is,
$u(a)$ is strictly increasing on $(-\infty,1]$ and strictly decreasing on $[1,\infty)$, and analogously for $u_V(a,v)$.\footnote{We adopt the convention that ``increasing'', ``larger than'', ``prefers'', etc., should be understood in the weak sense unless explicitly qualified by ``strict''.} Our main result (\autoref{prop:commitment_achievable_general} in \autoref{sec:general}) allows Proposer's utility $u$ to be any concave function but assumes that $u_V$ is quadratic loss.

A history in this game is a sequence of proposals. A strategy for Proposer is a function that assigns to every history a probability distribution over proposals, interpreted as the (possibly random) proposal Proposer makes given that all proposals in the history have been rejected. A strategy for Vetoer is a function that specifies for each history and each type the probability of accepting the last proposal. Our equilibrium concept is a standard version of Perfect Bayesian Equilibrium: both players play sequentially rationally and beliefs are updated by Bayes rule whenever possible---upon rejection of a proposal at any history, Proposer's belief about Vetoer's type is updated by Bayes rule if rejection has positive probability given Proposer's belief at that history. We also require, as usual, that Proposer's proposals do not (directly) affect his beliefs about Vetoer's type. 

Although our model formally has a single veto player, it can also be applied to settings in which Proposer has to secure approval from a committee of voters; so long as Proposer observes only whether his proposal passes or not, Vetoer can be interpreted as the median member of the committee. We elaborate in \autoref{sec:committee}.

\section{Two-Type Example}
\label{sec:twotypes}

This section presents an example to illustrate the logic of leapfrogging and how it benefits Proposer. The example has linear loss functions and a binary type distribution. Accordingly, for this section take
$$
u(a)=1-|1-a| \text{ \ and \ } u_V(a,v)=v -|v-a|,
$$
where the constants are determined by our normalization that both Proposer's and Vetoer's payoffs from the status quo (action $0$) are $0$. For simplicity, assume in this section that Proposer can only propose actions in $[0,1]$.
Suppose there are two Vetoer types, $l$ and $h$, and let $\mu_0$ be the prior probability of type $h$. We focus on the case where 
\begin{align}\label{Equation-Example}
    0<l<1/2<h<2l<1,
\end{align}
as it best illustrates the strategic issues at the core of our analysis. Proposer's first best---i.e., his optimum under complete information subject to Vetoer's approval---is action $1$ from type $h$ and action $2l$ from type $l$. The assumption that $h<2l$ implies that Vetoer of type $h$ prefers $2l$ to $1$ and so this first-best allocation cannot be implemented under incomplete information.

\paragraph{A Static Benchmark:} We begin our analysis with a useful benchmark. Consider a static (one-period) problem in which Proposer selects a menu of actions from which Vetoer can choose (if she opts to not exercise her veto);
equivalently, Proposer offers a deterministic mechanism or delegation set. In this problem, Proposer’s linear loss utility implies that he either pools both types with the singleton menu $\{2l\}$ or separates them using the menu $\{a^*,1\}$, where $a^*:=2h-1$ makes type $h$ indifferent between action $1$ and action $a^*$.\footnote{To see why optimal separation is via $\{a^*,1\}$, suppose separation is better than pooling and allocation $\{a^l,a^h\}$ with $a^l<a^h$ is an optimal separating allocation. It must be that $a^h>2l$; otherwise, pooling on $2l$ would be strictly better for Proposer. Hence, $a^l<h$; otherwise, both types would strictly prefer $a^l$. Consequently, each type $i\in\{l,h\}$ receives $a^i$. Incentive compatibility (IC) implies $a^l\leq 2h-a^h$; if this inequality is strict, raising $a^l$ a little preserves IC and is strictly profitable for Proposer. So $a^l=2h-a^h$, and it follows that only $a^h=1$ (which implies $a^l=a^*$) maximizes Proposer's payoff.} Separation is optimal whenever $\mu_0>\mu^*$, where $\mu^*$ is defined by
\begin{equation}
u(2l)=(1-\mu^*)u(a^*)+\mu^*u(1),
\label{eq:eg-sepoptimal}
\end{equation}
and pooling is optimal otherwise.
We refer to $\max\{u(2l),(1-\mu_0)u(a^*)+\mu_0 u(1)\}$ as Proposer's \emph{delegation payoff}. 

It is straightforward that when players are patient, Proposer can achieve approximately the delegation payoff in our sequential bargaining game if he could commit to a strategy.\footnote{Our analysis in \autoref{sec:general} shows that under certain conditions, the delegation payoff is in fact an upper bound on Proposer’s payoff in the dynamic game, even with commitment power. But those conditions ensure that delegation---a deterministic mechanism---is optimal in the static problem among stochastic mechanisms, which is not true in this example because of Vetoer's linear loss utility and discrete types. See also \autoref{fn:delayhelpsineg}.} But can Proposer achieve (approximately) the delegation payoff without commitment power?

\paragraph{The Sequential Rationality Problem:}  The difficulty when separation is optimal is that of \emph{Coasian dynamics}, which suggest the impossibility of screening Vetoer types when players are patient \citep*[e.g.,][]{FLT:85,GSW:86}, given that type $h$ prefers $l$'s allocation to her own. Specifically, if Proposer secures agreement initially (even with only high probability) from type $h$ on an action close to $1$, sequential rationality then impels him to offer $2l$ to reach an agreement immediately with type $l$. But anticipating the offer of $2l$, a patient type $h$ would not accept the initial high action. Indeed, it can be shown that in any equilibrium in which the on-path sequence of offers is decreasing---which guarantees that agreement is first secured with type $h$---Proposer's payoff at the patient limit is no higher than from pooling both types on action $2l$.  
This payoff is strictly below, and possibly far from, the delegation payoff when separation is optimal.

\paragraph{The Leapfrogging Solution:} Our key insight is that Coasian dynamics can be negated by \emph{leapfrogging}, i.e., making an offer that is accepted by the low type and rejected by the high type. Specifically,  
Proposer can first propose an action close to $a^*$ that is accepted only by type $l$. Upon rejection, Proposer credibly offers action $1$ ever after. In other words, leapfrogging uses a low action to first target the low type so that Proposer can subsequently extract a high action from the high type; crucially, at the latter stage, Proposer is no longer constrained by sequential rationality to moderate his offer if it is rejected. We highlight that it is Vetoer's single-peaked preferences that permit offers that type $l$ is willing to accept but type $h$ is not. 

We now make precise how Proposer can exploit leapfrogging with a succinct argument that presumes equilibrium existence. We argue that if separation is optimal, there is an equilibrium in which Proposer achieves  approximately the delegation payoff, at least. (Here and subsequently, we sometimes leave implicit that statements should be understood as holding for large $\delta$.) Let $a^{\delta}:=\delta a^*=\delta(2h-1)$ be the action below $h$ that makes type $h$ indifferent between obtaining action $1$ in the next period 
and obtaining action $a^{\delta}$ in the current period. Assume we are given an equilibrium. Modify that equilibrium to obtain a new equilibrium with strategy profile $\sigma$ and beliefs $\mu$ as follows:
\begin{enumerate}
\item \label{eg:nonconstr1}
if Proposer offers $a^{\delta}$ in the first period, type $l$ accepts and type $h$ rejects. After a first-period rejection of $a^\delta$, Proposer's belief assigns probability 1 to type $h$, and so he proposes $1$ in all future periods; in these periods, type $h$ accepts any proposal in $[a^\delta,1]$ and rejects all others, and type $l$ accepts any proposal in $[0,2l]$ and rejects all others;

\item \label{eg:nonconstr2} if Proposer offers $a\neq a^{\delta}$ in the first period, continuation play follows the original equilibrium;

\item \label{eg:nonconstr3} in the first period, Proposer chooses a proposal that maximizes his expected payoff.\footnote{We can assume a maximizer exists: if one doesn't, it must be that in the original equilibrium it is optimal for Proposer to choose $a^\delta$ in the first period, with a payoff larger than $(1-\mu_0)u(a^\delta)+\delta \mu_0 u(1)$; so the original equilibrium itself yields at least approximately the delegation payoff.}

\end{enumerate}

Point \ref{eg:nonconstr1} above implies that we have an equilibrium in the continuation game after a first-period proposal of $a^\delta$ is rejected. It follows from Points \ref{eg:nonconstr2} and \ref{eg:nonconstr3} that $(\sigma,\mu)$ is an equilibrium. In this equilibrium, either Proposer leapfrogs by offering $a^\delta$ in the first period which is accepted by type $l$, followed by action $1$ being accepted by type $h$ in the second period, or Proposer obtains an even higher payoff by proposing something different in the first period. When $\delta$ is close to $1$, $a^{\delta}$ is close to $a^*$ and Proposer's equilibrium payoff is close to the delegation payoff or even higher.

When separation is optimal, this argument shows that the \emph{option} to leapfrog yields Proposer approximately his delegation payoff or higher. But it does not establish that leapfrogging actually occurs, and it presumes equilibrium existence. We now turn to a full-fledged equilibrium construction that features leapfrogging; the construction also describes an equilibrium when pooling is optimal. 

\begin{proposition}
\label{prop:twotypes}
When $\delta$ is large, for any $\mu_0$ there is an equilibrium in which Proposer's payoff is approximately his delegation payoff.\footnote{More precisely: letting $u^d$ denote the delegation payoff, for all $\epsilon>0$ there is $\underline \delta<1$ such that for any $\delta>\underline \delta$ and for all $\mu_0$, there is an equilibrium in which Proposer's payoff is at least $u^d-\epsilon$.} 
In particular, there exist $\mu^\delta$ and $\bar\mu^\delta$, with $0<\mu^*<\mu^\delta<\bar\mu^\delta<1$, such that at $(\mu_0,\delta)$ there is an equilibrium  with on-path behavior as follows:
\begin{enumerate}[label={(\alph*)}]
\item \label{prop:twotypes1} (Skimming.) If $\mu_0 < \mu^\delta$, Proposer offers a finite sequence of actions that decreases to $2l$. Each offer strictly higher than $2l$ is accepted with positive probability by type $h$ and rejected by $l$. 

\item \label{prop:twotypes2} (Leapfrogging.) If $\mu_0 \in \left( \mu^\delta, \bar \mu^\delta\right)$, Proposer offers action $a^\delta$ in the first period, which is accepted by type $l$ and rejected by $h$; in the second period Proposer offers action $1$, which is accepted by type $h$.

\item \label{prop:twotypes3} (Delayed Leapfrogging.) If $\mu_0 > \bar \mu^\delta$, Proposer offers action $1$ in the first period, which is accepted with positive probability by type $h$ and rejected by $l$; in the second period Proposer randomizes between skimming and leapfrogging (parts \ref{prop:twotypes1} and \ref{prop:twotypes2}, respectively).
\end{enumerate}
\end{proposition}

(All proofs of formal results are in the Appendices.)

Case \ref{prop:twotypes1} of \autoref{prop:twotypes} concerns low priors.  Here we construct a \emph{skimming equilibrium} in which Proposer begins with an offer exceeding $2l$ but compromises to lower actions following rejections. As $\delta\rightarrow 1$, Proposer's payoff converges to the pooling payoff, $u(2l)$, from the static benchmark; moreover, $\mu^\delta$ also converges to $\mu^*$, and so for all priors less than $\mu^*$, Proposer is obtaining approximately his delegation payoff. The skimming equilibrium adapts a construction that is standard in seller-buyer bargaining  (\citealp{Hart:89}; \citealp[pp.~409--10]{FT:91}). However, there are novel considerations in deterring Proposer from offering actions lower than $2l$. In our construction, the most attractive deviation is leapfrogging, wherein Proposer first offers $a^\delta$ to secure acceptance from type $l$ and then extracts action $1$ from type $h$. Such deviations are profitable when type $h$ is sufficiently likely, which explains why our construction is an equilibrium only for a low prior (whereas in seller-buyer bargaining, the analogous equilibrium exists for all priors because no buyer type would wait for a higher price). The threshold $\mu^{\delta}$ is the (lowest) belief at which Proposer is indifferent between skimming and leapfrogging. 

\autoref{prop:twotypes}\ref{prop:twotypes2} and \ref{prop:twotypes3} are the main cases of interest, because here the prior is such that separation is optimal in the static benchmark. 
In Case \ref{prop:twotypes2}, Proposer leapfrogs at the outset, securing action $a^\delta$ from type $l$ in the first period and then action $1$ from type $h$ in the second period. As $\delta \to 1$, $a^\delta \to a^*$ and Proposer obtains his delegation payoff. The challenge with supporting leapfrogging is ensuring that Proposer does not deviate to a high offer in the first period. Such a deviation (if accepted with sufficient probability by type $h$) would be profitable if the prior is too large. 
The precise threshold $\bar \mu^\delta$ is determined by Proposer's indifference between leapfrogging and the most attractive deviation, which is an offer of $1$. In equilibrium this offer is accepted by type $h$ only with some probability, which brings Proposer's belief upon rejection down to the threshold $\mu^\delta$ described in the previous paragraph, so that Proposer then randomizes between skimming and leapfrogging in a manner that justifies $h$'s randomization. 
The full construction of the leapfrogging equilibrium is fairly involved; \autoref{Figure-Leapfrogging} summarizes, with details provided in the formal proof.

\begin{figure}[t]
\centering
{	\begin{tikzpicture}[scale=.81]
        \coordinate (zero) at (0,0);
        \coordinate (one) at (20,0);
        \coordinate (h) at (12,0);
        \coordinate (2l) at (14,0);
        \coordinate (adelta) at ($2*(h)-(one)-(1,0)$);
        \coordinate (abardelta) at ($(2l)+(2,0)$);
        
        \coordinate (ticklow) at (0,.-.2);
        \coordinate (tickhigh) at (0,.5);
        \coordinate (label) at (0,-.6);
        
		\draw[ultra thick] (zero) -- (one);
		\draw[gray] ($(adelta)+(ticklow)$)--($(adelta)+(tickhigh)$);
		\node at ($(adelta)+(label)$) {$a^\delta$};
		\draw[gray] ($(h)+(ticklow)$)--($(h)$);
		\node[blue] at ($(h)+(label)$) {$h$};
		\draw[gray] ($(2l)+(ticklow)$)--($(2l)+(tickhigh)$);
		\node at ($(2l)+(label)$) {$2l$};
	    \draw[gray] ($(abardelta)+(ticklow)$)--($(abardelta)+(tickhigh)$);
		\node at ($(abardelta)+(label)$) {$\bar{a}^\delta$};
		\draw[gray] ($(zero)+(ticklow)$)--($(zero)+(tickhigh)$);
		\node at ($(zero)+(label)$) {$0$};
		\draw[gray] ($(one)+(ticklow)$)--($(one)+(tickhigh)$);
		\node at ($(one)+(label)$) {$1$};
 		\node at ($.5*($(zero)+(tickhigh)+(adelta)+(tickhigh)$)$){{\footnotesize $I$}};
 		\node at ($.5*($(adelta)+(tickhigh)+(2l)+(tickhigh)$)$){{\footnotesize $II$}};
 		\node at ($.5*($(2l)+(tickhigh)+(abardelta)+(tickhigh)$)$){{\footnotesize $III$}}; 		
 		\node at ($.5*($(abardelta)+(tickhigh)+(one)+(tickhigh)$)$){{\footnotesize $IV$}}; 		
 	\end{tikzpicture}	
	}
\vspace{-0.1in}
\caption{
Proposer's first-period incentives in the equilibrium for \Cref{prop:twotypes}{\ref{prop:twotypes2}} and \ref{prop:twotypes3}.  Offers in Region $I$ (including $a^\delta$) are accepted only by type $l$; action $1$ is then offered and accepted by $h$. Offers in Region $II$ are accepted by both types. Offers in Region $III$ are accepted with some probability by  $h$ and rejected by $l$; rejection leads to a belief lower than $\mu^\delta$, whereafter there is a (suitably randomized) skimming equilibrium. Action $\bar a^\delta$ makes type $h$ indifferent between accepting $\bar a^\delta$ now and waiting one period to play \Cref{prop:twotypes}\ref{prop:twotypes1}'s skimming equilibrium under belief $\mu^\delta$. Offers in Region $IV$ are accepted by $h$ with some probability and rejected by $l$; rejection leads to belief $\mu^\delta$, whereafter Proposer mixes between skimming and leapfrogging. For any prior $\mu_0>\mu^\delta$, Proposer's optimal offer is either $a^\delta$ or $1$. Belief $\bar \mu^\delta$ is defined by Proposer's indifference between these two offers. Hence $\mu_0\in (\mu^\delta,\bar \mu^\delta)$ leads to leapfrogging (\Cref{prop:twotypes}\ref{prop:twotypes2}), whereas $\mu_0 > \bar \mu^\delta$ leads to a positive probability of delayed leapfrogging (\Cref{prop:twotypes}\ref{prop:twotypes3}).
}

\label{Figure-Leapfrogging}
\end{figure}
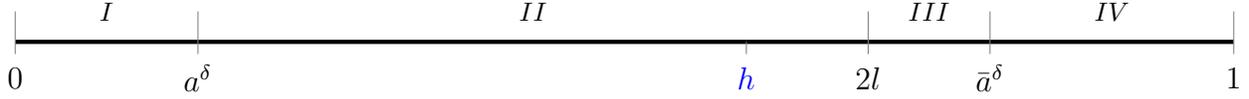

Finally, \autoref{prop:twotypes}\ref{prop:twotypes3} concerns the case of high priors, where leapfrogging from the outset cannot be sustained due to Proposer's strong incentive to secure agreement in the first period with the high type on a high action. 
Instead we have delayed and only probabilistic leapfrogging. As foreshadowed in the previous paragraph, now Proposer actually offers action $1$ in the first period, which is accepted by type $h$ with positive probability; upon rejection, Proposer randomizes in the second period between skimming and leapfrogging. Since Proposer is indifferent in the second period, his payoff is as if he always leapfrogs then, and his payoff therefore converges to the delegation payoff as $\delta \to 1$.

It is worth noting that although Cases \ref{prop:twotypes2} and \ref{prop:twotypes3} of \autoref{prop:twotypes} yield Proposer an identical payoff at the patient limit, both cases remain relevant even at that limit: $\lim_{\delta\to 1} \mu^\delta<\lim_{\delta\to 1} \bar \mu^\delta<1$ (see \autoref{fn:mubardelta_limit} in the appendix).
Moreover, since Proposer's delegation payoff becomes arbitrarily close to his complete-information payoff as $\mu_0\to 1$, \autoref{prop:twotypes} implies that there are equilibria in which Proposer's payoff at the patient limit is continuous in the prior even when the probability of type $l$ vanishes.\footnote{More precisely: $\lim_{\mu_0 \to 1} \lim_{\delta \to 1} U(\mu_0,\delta)=u(1)$, where $U(\mu_0,\delta)$ denotes Proposer's payoff in the equilibrium constructed in \autoref{prop:twotypes} for the belief $\mu_0$ and discount factor $\delta$.} By contrast, in seller-buyer bargaining, in any equilibrium (of the ``gap case''), the uninformed seller's payoff in the patient limit drops discontinuously when he ascribes any positive probability to the low-value buyer.

\paragraph{Limitations:} Although this example conveys the logic of leapfrogging and how Proposer can exploit it, there are two interrelated limitations. First, it is difficult to determine whether there are equilibria that are even better (or worse) for Proposer than that identified in \autoref{prop:twotypes}.
Second, while the delegation payoff provides a high target for Proposer, a more compelling benchmark is Proposer's payoff if he can commit to his strategy in the sequential bargaining game. Indeed, in this example dynamic commitments can achieve more than the delegation payoff.\footnote{\label{fn:delayhelpsineg}Let $t$ be the earliest period such that type $h$ prefers agreement on action $1$ in the first period to agreement on $2l$ in period $t$. 
If Proposer offers $1$ up until period $t-1$  and offers $2l$ from period $t$ on, then it is optimal for type $h$ to accept $1$ in the first period and for type $l$ to accept $2l$ in period $t$. For large $\delta$, $h$ is approximately indifferent: $2h-1\approx \delta^t (2h-2l)$, or equivalently, $(2h-1)\frac{l}{h-l}\approx \delta^t 2l$. It follows that Proposer's payoff from dynamic commitment is at least
 $\mu_0 u(1) + (1-\mu_0) \delta^t u(2l) \approx \mu_0 u(1) + (1-\mu_0) u(2h-1) \frac{l}{h-l}$. This latter expression is strictly larger than Proposer's payoff from the menu $\{a^*,1\}$ because $a^*\equiv 2h-1$ and $\frac{l}{h-l}>1$ (as $2l>h$ by assumption).
That dynamic commitment  strictly improves on the delegation payoff implies that the optimal static mechanism in this example must be stochastic (see \Cref{lemma:static_provides_upper_bound} below).
}  The following section addresses these issues by identifying assumptions within our general model such that Proposer (approximately) achieves his dynamic commitment payoff in an equilibrium.

\section{General Analysis}

\label{sec:general}

We hereafter assume Proposer's utility function $u(a)$ is concave and Vetoer's is 
\[u_V(a,v)=-(v-a)^2+v^2,\]
which is the standard quadratic loss function with our normalization that Vetoer's payoff from the status quo is $0$.
We also assume Vetoer's type is distributed according to a cumulative distribution $F \in \mathcal F$, where $\mathcal F$ is the set of distributions with interval support that admit a density that is bounded away from both $0$ and $\infty$ on the support. We denote the support of $F$ by $[\uv,\ov]$.
For this section alone, we assume that $\overline v \leq 1$, i.e., Vetoer's ideal point is always lower than Proposer's. We do not view this restriction as critical; indeed, our equilibrium constructions in \autoref{sec:equilibria} dispense with it. Note that we allow for $\ov \leq 1/2$, which is tantamount to Proposer having monotonic preferences.

Vetoer's quadratic loss function assures single-crossing expectational differences (SCED) as defined by \citet{KLR:19}: for any two lotteries over time-stamped actions---pairs $(a,t)$ representing agreement on action $a$ at time $t$, with $t=\infty$ capturing no agreement---their expected utility difference is single crossing in Vetoer's type $v$.\footnote{This is because the utility from any lottery over time-stamped actions is $-\E_{(a,t)}[\delta^t a^2]+2v\E_{(a,t)}[\delta^t a]$, which is linear in $v$. More generally, if $u_V(a,t)$ has SCED for non-time-stamped action lotteries (i.e., lotteries over actions within single period), then it follows from \citet*[Theorem 4]{KLR:19} that SCED will also hold over time-stamped action lotteries. We assume quadratic loss because of some additional tractability, but believe that our results would extend under SCED with weaker assumptions such as smoothness and concavity around the ideal point.} This single-crossing property will play an essential role because it guarantees ``interval choice'' \citep*[Theorem 1]{KLR:19}: given any Proposer strategy, at every history the set of types that find it optimal to accept the current offer is an interval.

\subsection{A Static Problem}
\label{sec:staticproblem}
We define an auxiliary static mechanism design problem that will turn out to provide
a tight upper bound on payoffs in the dynamic game. In this auxiliary problem, a (direct, stochastic) mechanism assigns each type a lottery over actions. Formally, a \emph{mechanism} $m$ is a measurable function $m:[\underline v,\overline v]\rightarrow M_0(\R)$, where $M_0(\R)$ is the set of probability distributions on $\R$ with finite expectation and finite variance. For notational convenience we write $m(v)=a$ when $m(v)$ puts probability $1$ on action $a$ and also extend the domain of Proposer's utility $u$ to include lotteries: $u(m(v)) := \E_{m(v)}[u(a)]$.   
A mechanism $m$ is \emph{incentive compatible} if every Vetoer type $v$ prefers $m(v)$ to $m(v')$ for all $v'$. It is \emph{individually rational} if every type $v$ prefers $m(v)$ to action $0$. Let $\mathcal S$ denote the set of incentive compatible and individually rational mechanisms.\footnote{More precisely, any $m \in \mathcal S$ must also be such that $v\mapsto \E_{m(v)}[u(a)]$ is integrable.}  Proposer's \emph{static problem} is:
\begin{align*}\label{e:static_problem}
\max_{m\in \mathcal{S}} \int u(m(v)) \mathrm dF(v). \end{align*}
We denote Proposer's maximum value by $U(F)$.

Any incentive compatible and individually rational mechanism that assigns every type a deterministic action can be implemented as a \emph{delegation set}: Proposer chooses a subset $A\subseteq \mathbb{R}$ and Vetoer is allowed to pick any action in $A\cup\{0\}$. We say that an \emph{interval delegation set} is optimal if a solution to the static problem can be implemented by delegating an interval $[c^*,1]$ for some $c^*\in[0,1]$. Our analysis below assumes environments in which such simple mechanisms are optimal. That is, we maintain hereafter:

 \begin{assumption}\label{ass:interval_optimal}
For some $c^*\in [0,1]$, an interval delegation set $[c^*,1]$ solves Proposer's static problem.
 \end{assumption}

The static problem has been studied by \citet*{KKVW:21}. Among other things, they motivate interval delegation and investigate when it is optimal. Their Corollary 3 establishes that sufficient conditions for  \autoref{ass:interval_optimal} are that Proposer's utility $u$ is a linear or quadratic loss function (or a combination thereof) and Vetoer's type density $f$ is logconcave.\footnote{While that paper maintains some assumptions on the type distribution that we don't assume, those assumptions are not needed for its sufficient conditions for optimality of interval delegation. We also note that the logic of Corollary 1 in that paper implies that the interval delegation set $[\max\{0,2\uv\},1]$ is an optimal mechanism if $f$ is decreasing on $[\max\{0,\uv\},\ov]$, given only that $u$ is concave.} Many commonly used distributions have logconcave densities \citep{BagnoliBergstrom:05}.

\subsection{An Upper Bound on the Commitment Payoff}

In the static problem, Proposer screens different Vetoer types by exploiting their heterogeneous preferences over (distributions of) actions within a single period. In our dynamic environment, delay is an additional screening instrument. Nevertheless, Proposer can do no better in the dynamic game \emph{even if he could commit to his strategy}:

\begin{lemma}\label{lemma:static_provides_upper_bound}
There is no Proposer strategy and Vetoer best response that yield Proposer a payoff strictly higher than $U(F)$.
\end{lemma}

The idea behind this result is straightforward, and familiar in the seller-buyer bargaining literature \citep[e.g.,][]{AD:89-Direct}: the outcome of any Proposer strategy and Vetoer best response can be replicated by a mechanism in the static problem. To elaborate, any Proposer strategy and Vetoer best response induce, for each Vetoer type, a probability distribution over agreements on time-stamped actions. We can transform any such distribution into a static lottery by mapping an agreement on action $a$ in period $t$ into a static lottery that gives action $a$ with probability $\delta^t$ and action $0$ with remaining probability. This transformation is payoff equivalent for Proposer and all Vetoer types. Therefore, the static mechanism induced by transforming each type's equilibrium distribution is incentive compatible and individually rational because Vetoer is playing a best response in the game, and the mechanism delivers Proposer the same payoff as in the game.

We highlight that while it is crucial that the static problem allow for stochastic mechanisms, the argument for \autoref{lemma:static_provides_upper_bound} does not require any assumption on either player's preferences beyond discounted expected utility 
with a common discount factor. Furthermore, the argument only uses the distribution of agreement times and actions for each type and the requirement that Vetoer is best responding to Proposer, nothing more about the game form. It follows that the static problem provides an upper bound on Proposer's commitment payoff in the dynamic game even if Proposer could, in any period, offer a menu of (possibly stochastic) actions, allow Vetoer to send cheap-talk messages, or engage in other complex protocols. Indeed, any incentive compatible and individually rational mechanism that assigns each type a lottery over time-stamped actions yields Proposer a payoff at most $U(F)$.

\subsection{Obtaining the Commitment Payoff without Commitment}
\label{sec:mainresult}

In light of \autoref{lemma:static_provides_upper_bound}, we say that \emph{Proposer can achieve approximately his commitment payoff for a belief $F'$} if given the belief $F'$ (at some history), for every $\epsilon>0$ there is $\underline \delta<1$ such that for all $\delta>\underline\delta$, there is a (continuation) equilibrium in which Proposer's payoff is at least $U(F')-\epsilon$. For brevity, we say that \emph{Proposer can achieve approximately his commitment payoff} if he can approximately achieve his commitment payoff for the prior $F$.\footnote{To be clear: conceptually, by ``commitment payoff'' we have in mind Proposer's payoff if he could commit to a strategy in the dynamic game. But operationally, we refer to the static problem's payoff $U(F)$ as the commitment payoff because of \autoref{lemma:static_provides_upper_bound}, our focus on large $\delta$, and \autoref{prop:commitment_achievable_general} below.} Our main result, \autoref{prop:commitment_achievable_general} below, presumes:
\begin{equation}
\label{eqmexists}
\text{An equilibrium exists for all $\delta$ and all beliefs in $\mathcal F$.}
\tag{EqmExists}
\end{equation}
We view this presumption as benign, and we provide reasonable sufficient conditions for equilibrium existence in \autoref{sec:equilibria}. In particular, it is sufficient that $\uv \leq 0$, i.e., some Vetoer types prefer the status quo to any action Proposer prefers.

\begin{theorem}\label{prop:commitment_achievable_general}
Suppose \ref{eqmexists}.
Proposer can achieve approximately his commitment payoff.
\end{theorem}

Together,  \autoref{lemma:static_provides_upper_bound} and \autoref{prop:commitment_achievable_general} imply that, when players are patient, there are equilibria in which Proposer suffers (almost)
no loss from the inability to commit in the dynamic game. In particular, Proposer is not harmed by the ability to make sequential proposals; in fact, whenever the optimal delegation set has $c^*<1$, Proposer strictly benefits from that ability, as the outcome from that delegation set cannot be replicated with a single proposal.  
Moreover, Proposer's gain from the ability to offer a menu of actions, rather than a single action, in each period vanishes as $\delta \to 1$.

\autoref{prop:commitment_achievable_general}'s conclusion may be best appreciated when $c^*>\max\{0,2\uv\}$, say $0<2\uv<c^*$. In that case the result contrasts with the negative conclusion from Coasian dynamics: intuitively, if Proposer were to continually compromise starting from a high offer, sequential rationality would drive offers all the way down to $2\uv$; it would not be credible for Proposer to stop at $c^*$.

An intuition one might proffer for \Cref{prop:commitment_achievable_general} is that, when $\delta \approx 1$, Proposer can begin with an offer of action $0$---leapfrog---and then offer a decreasing sequence of actions along a fine grid of $[c^*,1]$. Vetoer's best response would be to accept the offer of $0$ if her type is in $[0,c^*/2]$, and otherwise accept an offer in $[c^*,1]$, resulting in approximately the same outcome as the optimal delegation set $[c^*,1]$. This intuition is incomplete because Proposer must be incentivized to offer $0$ initially, and even thereafter, it is not clear that he would be willing to follow the decreasing sequence of offers. We are able to substantiate this intuition in \autoref{sec:highpayoffeqm} under some conditions, by exploiting equilibrium payoff multiplicity to deter any Proposer deviations. 
Instead, we pursue a different approach to prove \autoref{prop:commitment_achievable_general} that does not rely on equilibrium payoff multiplicity and highlights the power of Proposer's \emph{option} to leapfrog. It is this argument that we sketch in the remainder of this subsection.

Our first step is to derive a ``conditional optimality'' property of interval delegation: given the assumption that delegation set $[c^*,1]$ is an optimal static mechanism for the prior type distribution $F$, it is also optimal for certain conditional distributions. To state the result, let $F_{[v_1,v_2]}$ denote the conditional distribution of $F$ given $v\in[v_1,v_2]$, for any $v_1,v_2\in [\underline{v},\overline{v}]$ with $v_1\leq v_2$.

\begin{lemma}\label{lemma:opt_delegation_set}
 The delegation set $[c^*,1]$ solves Proposer's static problem for any belief $F_{[c,c']}$ with $c\le c^*/2\le c^*\le c'$.
\end{lemma}

The lemma owes to SCED of Vetoer's utility and the optimal static mechanism being interval delegation, rather than just an arbitrary delegation set. The proof uses these properties to establish that if some mechanism outperforms delegation set $[c^*,1]$ for any of the relevant truncated beliefs, then augmenting that mechanism by adding an interval of high actions yields a mechanism that also outperforms $[c^*,1]$ for the original belief.

\autoref{lemma:opt_delegation_set} says, in particular, that delegation set $[c^*,1]$ is an optimal mechanism for the belief $F_{[\underline{v},c^*]}$ and that it remains
optimal for the belief $F_{[c^*/2,c^*]}$ that is induced if Proposer leapfrogs and obtains agreement from all types below $c^*/2$. We use these properties to next establish \autoref{prop:commitment_achievable_general} for the special case in which Proposer's belief is $F_{[\underline{v},c^*]}$.

\begin{lemma}\label{lemma:induction_base}
Suppose \ref{eqmexists}. Proposer can achieve approximately his commitment payoff for belief $F_{[\underline v,c^*]}$.
\end{lemma}

 The proof deduces an equilibrium in which Proposer has an option to leapfrog that guarantees him approximately the commitment payoff, analogous in spirit to the logic given before \autoref{prop:twotypes}. In the equilibrium, Proposer has the option to follow a path in which he first proposes action $0$, which will be accepted by all types below $c^*/2$, and then proposes action $c^*$, which will be accepted by all types above $c^*/2$. When players are patient this path yields Proposer approximately the same payoff as in the static problem because the delegation set $[c^*,1]$ is outcome-equivalent to $\{c^*\}$ under the belief $F_{[\underline{v},c^*]}$. On this path, Proposer's sequential rationality in the second period with belief $F_{[c^*/2,c^*]}$ is assured by \autoref{lemma:static_provides_upper_bound} and \autoref{lemma:opt_delegation_set}. Sequential rationality for Vetoer after both the initial proposal of $0$ and the subsequent proposal $c^*$ is because a rejection of $c^*$ in the second period would lead Proposer to put probability $1$ on type $c^*$ and make subsequent proposals that are larger than $c^*$, and hence worse for Vetoer regardless of her type in $[\underline v,c^*]$.\footnote{While it is weakly dominated for Vetoer to accept a proposal of $0$, we use action $0$ because of the continuum action space. There are discretizations of the action space in which Proposer's leapfrogging option can be constructed using a strictly positive action instead of $0$.}

\autoref{lemma:induction_base} serves as the base step for an inductive proof of \autoref{prop:commitment_achievable_general}. Specifically, we show that if Proposer can achieve approximately his commitment payoff for belief $F_{[\underline v,c']}$ for some $c'\ge c^*$, then there is a neighborhood of $c'$ such that for any $c''$ in this neighborhood, the commitment payoff can also be achieved approximately for belief $F_{[\underline v,c'']}$.\footnote{This explanation is heuristic; the formal proof ensures that for any $\epsilon>0$, for all large enough $\delta<1$, the induction can traverse the set of types with Proposer obtaining a payoff at least $U(F)-\epsilon$.}

Here is the idea for the inductive step.  Consider the action $a'>c'$ that makes type $c'$ indifferent between accepting $a'$ in the current period and playing a putative continuation equilibrium 
with belief $F_{[\underline v,c']}$ that gives Proposer approximately his commitment payoff under that belief. Presuming this continuation if $a'$ is rejected, it is optimal for types below $c'$ to reject $a'$ because SCED implies that they obtain a higher payoff from using the strategy of type $c'$ in the continuation equilibrium. On the other hand, there is a neighborhood of types above $c'$ within which it is optimal to accept $a'$ because (i) discounting implies that types in a neighborhood of $a'$ prefer accepting $a'$ to receiving even their ideal action in the next period, and (ii) SCED implies that the set of types willing to accept any proposal is an interval. Now suppose Proposer's belief is $F_{[\underline v,c'']}$ for $c''$ strictly larger than but sufficiently close to $c'$. It follows that the belief $F_{[\underline v,c']}$ and the continuation equilibrium we hypothesized is self-fulfilling: anticipating this continuation leads to $a'$ being rejected by precisely the set of types $[\underline v,c']$. Moreover, action $a'$ is an option that assures Proposer approximately his commitment payoff: conditional on rejection by types less than $c'$, the continuation results in approximately the commitment payoff given the conditional distribution, whereas every type $v \in (c',c'')$ accepts proposal $a'>c''$  which  is larger than the action $v$ that Proposer gets from type $v$ in the static problem for belief $F_{[\underline v,c'']}$ (by \autoref{lemma:opt_delegation_set}, given that $c''>c'\geq c^*$).

We highlight that our proof of \autoref{prop:commitment_achievable_general} uses a leapfrogging option to deduce a high-payoff equilibrium for Proposer without actually identifying his equilibrium strategy or the equilibrium outcome (i.e., the mapping from Vetoer types to time-stamped action distributions).\footnote{This is reminiscent of the approach used in the reputation literature \citep[e.g.,][]{FL:89,FL:92}, among other places, although the logic here is distinct. Unlike in those classic papers, we have two long-lived players, and there can also be equilibria in which Proposer obtains a low payoff (\autoref{prop:existence_decreasing_offers} below).} As explained above, the proof uses induction on beliefs of the form $F_{[\underline v,c]}$, exploiting the ``conditional optimality'' of the ex-ante optimal mechanism for such beliefs (\autoref{lemma:opt_delegation_set}). However, in a leapfrogging equilibrium, Proposer’s beliefs need not take only that form. But that is compatible with conditionally optimality of the ex-ante optimal mechanism---indeed, \autoref{lemma:opt_delegation_set} assures that the interval $[c^*,1]$ remains an optimal mechanism so long as Proposer's belief is of the form $F_{[c,c']}$ with $c\leq c^*/2\leq c^*\leq c'$. 
We will see in \autoref{sec:highpayoffeqm} that, under some conditions, there are leapfrogging equilibria in which Proposer's beliefs always have this form on the equilibrium path.

Moving beyond interval delegation, we do not know in general whether our proof strategy for \autoref{prop:commitment_achievable_general} can be used when the optimal mechanism is an arbitrary delegation set; what would be important for our approach is that the delegation set be a conditionally optimal mechanism for a suitable range of beliefs.

\subsection{Committee of Voters}\label{sec:committee}

Our analysis with a single Vetoer can be extended to situations in which a committee votes on Proposer's offer. For some odd number $N$, consider a committee of $N$ voters that aggregates votes via simple majority rule. Each voter $n\in\{1,\ldots,N\}$ has the utility function $u(a,v_n)$, where $v_n$ is her ideal point. Ideal points are drawn from some prior joint distribution, which need not be independent across voters. Each voter observes the realized vector $(v_1,\ldots,v_n)$, but Proposer does not. Crucially, Proposer also does not observe the vote profile in any period, only whether his offer passes. It does not matter whether the voters observe each others' votes.

Let $m:= (N+1)/2$ and let $F$ denote the distribution of the median (i.e., $m$-th highest) ideal point. We claim that so long as $u$ has SCED, every equilibrium of our Proposer-Vetoer two-player game with type distribution $F$  has an outcome-equivalent equilibrium of the committee game. Specifically, the committee game's equilibrium can be described as follows: (i) Proposer behaves just like in the two-player game; (ii) the realized median voter (i.e., the voter who realizes the $m$-th highest ideal point), say voter $m$, behaves just like Vetoer with type $v_m$; and (iii) at every history, every non-median voter votes for the current proposal if and only if she prefers it to the distribution of future agreements (time-stamped actions) induced in the two-player game if Vetoer has type $v_m$ and rejects at that history. Note that all voters behave ``sincerely'' or ``as if pivotal'' in the sense of voting at every history based on their comparison of the current offer with what will happen, in equilibrium, if the offer does not pass.

Here is why the above strategies form an equilibrium of the committee game. Without loss, assume the realized vector of ideal points has $v_1\leq \cdots \leq v_n$. The key observation is that all voters share a common belief about the distribution of future agreements (since $v_m$ is known to all voters), and so SCED assures that the set of voters who have the same preference as the median voter $m$ to accept (or reject) the current offer includes either $\{1,\ldots,m\}$ or $\{m,\ldots,N\}$. Hence, the median voter is always decisive, and all voters are playing sequentially rationally if the median voter is. Since Proposer only observes whether his offer was accepted or rejected, and the median voter behaves just like in the two-player game, it follows that Proposer is behaving sequentially rationally. Finally, being decisive, the median voter is clearly also playing sequentially rationally.

\section{Equilibrium Constructions and Multiplicity}\label{sec:equilibria}

This section constructs two equilibria: a leapfrogging equilibrium that yields Proposer approximately his commitment payoff, and a skimming equilibrium that can yield him a significantly lower payoff. Both constructions require some (plausible) assumptions on the support of the type distribution. Under those assumptions, they settle the equilibrium existence presumed by \Cref{prop:commitment_achievable_general}. Moreover, we also establish a sense in which leapfrogging is necessary to achieve the commitment payoff. Unlike in \autoref{sec:general}, we now permit the upper bound of the type distribution, $\ov$, to be larger than $1$.

\subsection{A Skimming Equilibrium}
\label{sec:skimmingequilibrium}

We first construct a \emph{skimming equilibrium}, which we define, following standard practice \citep[e.g.,][p.~407]{FT:91}, as an equilibrium in which any on-path non-negative offer is accepted by an upper set of Vetoer types.\footnote{We qualify the upper-set acceptance to hold only for (i) non-negative offers and (ii) on-path offers. Point (i) is needed because of Vetoer's single-peaked preferences: if a strictly negative offer is accepted by any remaining types, the acceptance set cannot be an upper set since high types prefer the status quo. Regarding (ii), we could use the stronger definition that includes off-path offers---and our construction in \autoref{prop:existence_decreasing_offers} satisfies that requirement---but restricting to on-path offers strengthens \autoref{prop:leapfrogging_necessary} in \autoref{sec:leapfroggingnecessary} and its implication that leapfrogging is necessary for the commitment payoff.}
This skimming equilibrium shows that a Coasian intuition does have some merit in our setting, which makes it more striking that the commitment payoff can also be achieved. Furthermore, we establish that Proposer’s payoff in our skimming equilibrium converges in the patient limit to that of \emph{full delegation}, i.e., of simply allowing Vetoer to choose her preferred action in $[2\uv^+,1]$, where $\uv^+ := \max\{0,\uv\}$.\footnote{In other words, full delegation is delegation of the interval $[c,1]$ where $c=0$ if $\uv\leq 0$ and $c=2\uv$ if $\uv \in (0,1/2)$. Note that we ignore here, and in the rest of \autoref{sec:equilibria}, the case of $\uv>1/2$; it is uninteresting because there is trivially a skimming equilibrium in which Proposer obtains his ideal point by offering $1$ at every history. Nonetheless, all our statements hold even if $\uv>1/2$ so long as in that case one interprets the notation $2\uv^+$ to mean $1$.}
It follows that there can be a substantial multiplicity in bargaining outcomes.

To state the result, define 
$${\underline U(F)} := \int_{\uv}^{2\uv^+} u(2\uv^+)\mathrm dF(v) + \int_{2\uv^+}^{1} u(v)\mathrm dF(v)+\int_{1}^{\max\{\ov,1\}} u(1)\mathrm dF(v)$$
as the static payoff from delegation set $[2\uv^+,1]$. In this mechanism all Vetoer types below $2\uv^+$ are pooled at action $2\uv^+$,  all types in $[2\uv^+,1]$ obtain their ideal points, and all higher types are pooled at $1$.
\begin{proposition}\label{prop:existence_decreasing_offers}
If either $\uv \le 0$ or $ \ov \le 1/2$, then there is a skimming equilibrium. As $\delta\rightarrow 1$, Proposer's payoff in this equilibrium sequence converges to ${\underline U(F)}$.
\end{proposition}

For any $\delta$, we construct a skimming equilibrium by adapting the approach used in seller-buyer bargaining \citep*[e.g.,][]{GSW:86,AD:89}. Suppose that Proposer's belief at any history is a right-truncation of his prior, i.e., the set of remaining Vetoer types is $[\uv,v]$ for some $v$. The highest remaining type can be used as a state variable for dynamic programming to find Proposer’s optimal sequence of decreasing offers, with a constraint that each subsequent state must be induced by Vetoer’s best response of accepting an offer if and only if she prefers it to the discounted payoff from accepting the subsequent offer. \autoref{def:support_eq} in \appendixref{app:decreasing} makes this program precise. As we discuss there, single-peaked Vetoer preferences introduce some differences in how we formulate and tackle the program relative to seller-buyer bargaining.

A novel issue arises in verifying that there is an equilibrium corresponding to a solution to the aforementioned program: what happens if $\uv>0$ and Proposer deviates at some history to an offer in $[0,2\uv)$? The issue is salient because, unlike in seller-buyer bargaining, leapfrogging could be attractive to Proposer. We use \autoref{prop:existence_decreasing_offers}'s hypothesis that $\ov\leq 1/2$ (given $\uv>0$)  to deter such deviations by stipulating that any such offer is accepted by all Vetoer types, which makes it unattractive to Proposer. It is optimal for Vetoer to accept these low offers because we specify Proposer’s belief after rejection to be degenerate on $\ov$, and accordingly Proposer’s future offers to perpetually be $2\ov$, which yields no surplus to any Vetoer type.\footnote{Our solution concept of Perfect Bayesian equilibrium allows for arbitrary beliefs after a rejection that has zero probability at that history. As such, even if $\ov>1/2$ (and $\uv>0$), strictly speaking one could assign the degenerate belief on $0$ after an unexpected rejection and have Proposer offer action $0$ ever after, which would also yield no surplus to all Vetoer types. 
We do not allow for such beliefs, instead requiring---as is conventional, and in the spirit of \citepos{KW:82} sequential equilibrium---that beliefs must always be supported in the support of the prior, $[\uv,\ov]$.} Both $\uv \leq 0$ and $\ov\leq 1/2$ are reasonable hypotheses: the former says that the status quo may be Pareto efficient; the latter is tantamount to Proposer having monotonic preferences over the set of actions that any Vetoer type would find acceptable.

Another distinction with seller-buyer bargaining is that, as $\delta \to 1$, Proposer's payoff in the skimming equilibrium converges to the full-delegation payoff ${\underline U(F)}$, rather than the payoff from all types accepting $2\uv^+$. On the one hand, our argument for why Proposer’s payoff in the limit cannot be larger than ${\underline U(F)}$ builds on ideas in that literature; roughly, a type $v>2 \uv^+$ would accept an offer strictly higher than $v$ only if there is a significant delay cost to waiting for a more attractive offer, but such a delay cost would make it attractive for Proposer to deviate and hasten agreement. On the other hand, a new observation owing to our setting is that Proposer's payoff cannot be lower than ${\underline U(F)}$ either: intuitively, because of her single-peaked utility, for any $\delta<1$  Vetoer will accept any proposal close enough to her ideal point; hence, as $\delta \to 1$, Proposer must do no worse in the skimming equilibrium than by compromising with an arbitrarily fine sequence of offers that traverses $[2 \uv^+,1]$.

In general, Proposer's payoff from the skimming equilibrium when players are patient, ${\underline U(F)}$, will be strictly less than his commitment payoff, $U(F)$; these payoffs coincide only when full delegation is an optimal mechanism, i.e., the $c^*$ threshold in \autoref{ass:interval_optimal} is precisely $2\uv^+$.
\citet*[Corollary 1]{KKVW:21} identify that full delegation is in fact optimal if the type density is decreasing on $[2\uv^+,1]$. 
Observe that when $\uv\leq 0$, the skimming equilibrium's payoff is a lower bound on Proposer's payoff from any equilibrium when players are patient; for, no equilibrium can yield Proposer a payoff strictly lower than from delegating the $[0,1]$ interval. It follows that if full delegation is optimal and $\uv\leq 0$, then when players are patient all equilibria must yield Proposer the commitment payoff.

Notwithstanding such cases, the general contrast in Proposer’s payoff between \Cref{prop:commitment_achievable_general} and \Cref{prop:existence_decreasing_offers} indicates the importance of equilibrium selection, which we interpret as norms, in veto bargaining. Which norm prevails in a given context may hold significant implications for whether Proposer suffers from an inability to commit to future offers. Moreover, in some environments---e.g., when Proposer prefers a single take-it-or-leave-it offer to full delegation---the norm can determine whether Proposer benefits from or is harmed by the ability to make multiple proposals. But in other environments---e.g., when $\uv\leq 0$ and Proposer prefers full delegation to a single offer---the ability to make multiple proposals benefits Proposer regardless of the norm. We highlight that both the sequential structure of bargaining and incomplete information are necessary for norms to matter in veto bargaining; in particular, \citet{Primo:02} shows that there is a unique equilibrium outcome absent incomplete information.

\subsection{A Commitment-Payoff Equilibrium}\label{sec:highpayoffeqm}

We now build on the previous subsection's skimming equilibrium to construct a leapfrogging equilibrium---one with leapfrogging on path---that delivers (approximately) Proposer's commitment payoff. The construction reveals how the dynamics of leapfrogging may play out, subject to a reasonable assumption that either $\uv \leq 0$ (i.e., the status quo may be Pareto efficient) or $\ov\leq 1/2$ (i.e., Proposer effectively has monotonic preferences), and that full delegation is not optimal. Note that if full delegation is optimal, then skimming achieves the commitment payoff (\autoref{prop:existence_decreasing_offers}).

\begin{proposition}
\label{prop:commitment_construction}
Suppose that either $\uv \leq 0$ or $\ov \leq 1/2$, and that full delegation is not optimal. There is a leapfrogging equilibrium in which Proposer achieves approximately his commitment payoff. In this equilibrium, Proposer first offers $0$, which is accepted if and only if $v\in (0,c^*/2)$; subsequently, Proposer offers a decreasing sequence of offers that culminates in $c^*$, with each offer accepted by an upper interval of remaining types.
\end{proposition}

In the equilibrium identified by \Cref{prop:commitment_construction}, Proposer begins by leapfrogging with an offer of $0$; if that offer is rejected, he knows that Vetoer's type is either below $0$ or above $c^*/2$. Naturally, he is only concerned with the latter possibility. So, upon the rejection of offer $0$, we are able to use essentially the same skimming construction as in \Cref{prop:existence_decreasing_offers}, but with the conditional distribution $F_{[c^*/2,\ov]}$. For large $\delta$, this implements a fine-grid sequence of decreasing offers down to $c^*$. As $\delta\to 1$, the overall outcome thus converges to that of Vetoer simply choosing (with no delay cost) her preferred action from the optimal delegation set $[c^*,1]$, or exercising her veto. 

Let us highlight a few points about the construction. First, for the reasons discussed after \autoref{prop:existence_decreasing_offers}, we use the hypothesis that either $\uv\leq 0$ or $\ov\leq 1/2$ to ensure validity of the skimming construction after offer $0$ has been rejected. Notably, then, \autoref{prop:commitment_construction} is valid even when $\ov>1$, so long as $\uv\leq 0$. Second, the equilibrium must incentivize Proposer in the first period to offer action $0$ rather than some higher action. This is ensured by stipulating that if Proposer deviates to action $a>0$ in the first period, continuation play follows that of the skimming equilibrium constructed in \Cref{prop:existence_decreasing_offers}. Such a deviation yields Proposer a payoff no more than (approximately) the payoff from full delegation, which is strictly less than the commitment payoff that is approximately achieved on path.

Third, although we view the leapfrogging-followed-by-skimming dynamics in \autoref{prop:commitment_construction} to be intuitive, we do not rule out other dynamics that also deliver approximately Proposer's commitment payoff. In particular, it is plausible that one may use the same approach to construct equilibria in which Proposer begins with some skimming, then leapfrogs with offer $0$, and then continues skimming again. There may also be other dynamics. Fourth, \autoref{prop:commitment_construction} crucially exploits equilibrium payoff multiplicity: we use a low-payoff skimming equilibrium to construct a high-payoff equilibrium. This approach is reminiscent of the ``reputational equilibria'' in \cite{AD:89}. By contrast, the logic we use to prove our main result, \Cref{prop:commitment_achievable_general}, does not leverage equilibrium payoff multiplicity; it would apply even if there is no skimming equilibrium and even if all equilibria yield Proposer a high payoff.\footnote{On the other hand, we noted at the end of \autoref{sec:mainresult} that it is not straightforward to extend the approach used in proving \autoref{prop:commitment_achievable_general} absent optimality of interval delegation (\autoref{ass:interval_optimal}). But given a low-payoff equilibrium, the logic underlying \autoref{prop:commitment_construction}'s construction ought to support a high-payoff equilibrium so long as some deterministic mechanism---even if not interval delegation---solves the static problem.}

\subsection{Is Leapfrogging Necessary?}\label{sec:leapfroggingnecessary}

We have highlighted leapfrogging as the driving force to achieve Proposer's commitment payoff, so long as full delegation is not optimal (in which case, by Section 5.1, skimming suffices).  In fact, leapfrogging is then more or less necessary:

\begin{proposition}\label{prop:leapfrogging_necessary}
Suppose that the essentially unique solution to the static problem is an interval delegation set that is not full delegation. Proposer's payoff in any skimming equilibrium is bounded away (across $\delta$) from the commitment payoff.
\end{proposition}

We view the assumption that the static problem has a unique solution (essentially---i.e., up to a set of types of measure $0$) as mild. That it is not full delegation is equivalent to $c^*>2\uv^+$. For instance, this inequality holds when $\uv \leq 0$, $u(\cdot)$ is affine on $[0,1]$, and Vetoer's type density $f$ is logconcave and attains a unique peak at some $v>0$.\footnote{An affine $u$ and logconcave $f$ ensure that interval delegation is optimal; $f$ having a unique peak at $v>0$ implies the interval's threshold is $c^*>0$. See \citet*{KKVW:21}.} Note that $\uv \leq 0$ assures existence of both a skimming equilibrium (\autoref{prop:existence_decreasing_offers}) and a commitment-payoff equilibrium (\autoref{prop:commitment_construction}).

The intuition for \autoref{prop:leapfrogging_necessary} is that for any large $\delta<1$, to achieve close to the commitment payoff, the outcome must be approximately that (i) Proposer reaches agreement with all types above $c^*/2$ on their preferred actions in $[c^*,1]$ without excessive delay, and (ii) all types below $c^*/2$ obtain the status quo (or some other actions only after significant delay). But if (i) happens in a skimming equilibrium, then eventually Proposer will be faced with, approximately, the type distribution $F_{[\uv,c^*/2]}$, in which event he will not find it optimal to induce (ii); he could profitably deviate to a fine-grid sequence of offers in $[0,c^*/2]$ that are accepted by most remaining positive types with virtually no delay cost. Note that this logic applies even if we are in the no-gap case ($\uv\leq 0$).

Subject to its conditions, \autoref{prop:leapfrogging_necessary} implies that any equilibrium that achieves approximately the commitment payoff must, with positive probability, have a leapfrogging offer $a\geq 0$ that is accepted by some low type and yet rejected by some higher type. In such an equilibrium, with positive probability, the sequence of on-path offers will not be decreasing: for, an upper set of types would accept the current offer if future offers are certain to be lower. Therefore, leapfrogging plays an indispensable role in yielding the commitment payoff.

\section{Related Literature}
\label{sec:literature}
\label{sec:discussion}

We now relate our work to some prior literature.

\paragraph{Veto Bargaining with Incomplete Information:} 

Existing work on sequential veto bargaining with incomplete information focuses on short horizons, typically two periods, and/or myopic Vetoer behavior (e.g., \citealt{RR:79}, \citealt{DR:92}, Chapter 4 of \citealt{Cameron:00}, \citealt{RosenthalZame2022}, \citealt{Chen2021}).\footnote{We highlight work that is most closely related to ours. But there have, of course, been studies on other aspects of veto bargaining with incomplete information. For example, \citet{Matthews:89} models veto threats, whereby Vetoer sends a cheap-talk message prior to Proposer making a take-it-or-leave-it offer. \citet{mccarty1997presidential} considers two-issue bargaining, wherein Vetoer may reject a proposal on one issue to influence proposals on the second issue. \citepos{groseclose2001politics} model of blame-game politics shows that in a three-player game, Proposer may make an offer that he knows Vetoer will reject in order to convince a third party (e.g., voters) that Vetoer has extreme preferences.} These analyses elucidate nicely some of the strategic forces, but 
either a short horizon 
or myopic Vetoer behavior precludes the potency of Coasian dynamics. The only exception to these approaches that we are aware of is the unpublished work of \cite{CE:94}, who study a long finite horizon with sophisticated players. All these authors, including \citeauthor{CE:94}, are interested in skimming equilibria. Our analysis shows that---unlike in seller-buyer bargaining---it is important to account for the possibility of leapfrogging because that can both invalidate a putative skimming equilibrium (recall the discussions of both \autoref{prop:twotypes}\ref{prop:twotypes1} and \autoref{prop:existence_decreasing_offers}) and lead to qualitatively different equilibria with higher Proposer payoff.

Recently, in a two-period model, \citet{Evdokimov:21} has emphasized what he views to be ``non-Coasian'' equilibria in veto bargaining. He studies committees in which voter preferences are determined by a binary state, analogous to our two-type example. Single-peaked voter preferences are important to his analysis, as they are to ours; however, our papers focus on distinct implications of single-peakedness, and the nature and import of our results are markedly different. To see that, consider his setting when a single vote is enough to overturn the status quo; it is effectively then as if Proposer faces a single vetoer. Here \citeauthor{Evdokimov:21} finds a unique equilibrium, which has skimming. Leapfrogging does not arise because of the combination of only two periods and his assumption that Proposer’s utility is globally increasing in the action.\footnote{An analog would be a two-period version of our \autoref{sec:twotypes} with the assumption that $h<1/2$. In that case, if type $l$ agrees first, then agreement in the second period with type $h$ has to be on action $2h$, which provides $h$ no surplus; so the only first-period action that can support leapfrogging is $0$, which turns out to be unsupportable for any prior. On the other hand, when either $h>1/2$ or there are more than two periods with $\delta<1$, arguments related to those for \autoref{prop:twotypes} can be used to conclude that leapfrogging is supportable for suitable priors.}   Instead, what \citeauthor{Evdokimov:21} deems non-Coasian are equilibrium outcomes in which, using our two-type notation from \autoref{sec:twotypes}, Proposer obtains utility that exceeds $u(2l)$ as $\delta\rightarrow 1$. He notes that such outcomes arise if $h>2l$. The reason is simply that type $h$ prefers some actions strictly above $2l$ to $2l$, and hence Proposer can guarantee a utility exceeding $u(2l)$ by first offering $h$ and then $2l$.  By contrast, we focused on arguably the more interesting case of $h<2l$, because that means separation cannot be achieved (when players are patient) with both types getting actions above $2l$. More generally, we do not take a stance on what the Coase Conjecture ought to mean in veto bargaining. Instead, our key contribution for two types and beyond is to unsheathe the leapfrogging implications of single-peaked preferences, which yield equilibria that have non-skimming dynamics and high Proposer payoffs. Furthermore, our main result (\autoref{prop:commitment_achievable_general}) is substantially stronger than just comparing with a single take-it-or leave it offer, which is \citepos{Evdokimov:21} benchmark.

\paragraph{Seller-Buyer Bargaining:}  In canonical models of seller-buyer bargaining in which the buyer is privately informed of his value, all equilibria feature skimming. \citet*{FLT:85} and \citet*{GSW:86} establish the Coase Conjecture: at the patient limit, the seller’s payoff is that of pricing at the lowest buyer valuation. More precisely, this holds in any equilibrium of the ``gap'' case (the gains from trade are bounded away from $0$) or in any ``stationary/weak Markov’’ equilibrium of the ``no gap'' case. Indeed, there is a unique equilibrium payoff for the seller in the gap case. By contrast, even in the gap case of our model (i.e., $\underline v>0$), Proposer can obtain his commitment payoff and there can be genuine payoff multiplicity. \cite{AD:89} show that in the seller-buyer no gap case, there also exists a non-stationary ``reputational equilibrium’’ in which the seller obtains his commitment payoff. This equilibrium preserves high prices by punishing the seller with Coasian low-payoff continuation play if he deviates. Our argument for Proposer’s commitment payoff is distinct; it owes to leapfrogging, which is ruled out by the skimming property of seller-buyer bargaining.\footnote{For a gap-case specification, \citet{DS:21} show that the Coasian outcome cannot be escaped even using arbitrary within-period mechanisms. In our setting, even if we allow for such mechanisms, it follows from the discussion in \autoref{sec:staticproblem} that our commitment payoff is still an upper bound; consequently, the equivalence between commitment and Proposer's best no-commitment equilibrium would prevail.}

\cite{board2014outside} show that when buyers have outside options, there is a unique equilibrium outcome and it yields a high seller payoff. The seller charges the static monopoly price---defined for the distribution of values net of the outside option---and all buyer types with lower net values immediately take their outside option. Since low types exit immediately, the seller can credibly stick to the monopoly price even upon rejection. In our analysis, leapfrogging also clears low types to subsequently credibly target high types. But our model has no outside options and it is Vetoer's single-peaked preferences that makes leapfrogging viable. Moreover, unlike in \citet{board2014outside}, low-payoff equilibria can coexist with the commitment-payoff equilibrium.\footnote{\cite{HL:17} and \cite{Fanning:21} highlight equilibrium multiplicity in seller-buyer models related to \cite{board2014outside}.} The idea that low agent types’ incentives to exit can allow a principal to obtain her commitment payoff also features in \cite{tirole2016bottom}. But there, unlike in our model, a reverse-skimming property holds, i.e., any equilibrium has ``positive selection’’ at every history.

Also related to our work are models in which the seller sells multiple varieties. \cite{Wang:98}, \cite{hahn2006damaged},  and \citet{Mensch:17} study bargaining when there is a choice of both quality and price (or effort and wage in a labor context). In these models, the seller or principal offers a menu in each period but cannot commit to future menus. The key finding is that the principal obtains his commitment payoff in the unique equilibrium. More recent developments include   \cite{nava2019differentiated}, who propose a multidimensional extension of the Coase Conjecture, and \citet{Peski:22}, who establishes payoff uniqueness in a broad class of bargaining protocols and mechanisms.\footnote{Although \citet{Peski:22} studies a single indivisible good, he allows for commitments to probabilistic trade, which is effectively the same as varieties.}
In our model, not only are transfers infeasible, but moreover Proposer can offer only a single action, rather than a menu, in each period. This hews to the standard approach in studying sequential veto bargaining, and seems appropriate for some non-market applications in politics and organizations. Nevertheless, we deduce equilibria that deliver Proposer's commitment payoff. It would be interesting to study whether allowing for menus eliminates the payoff multiplicity we find. Conversely, our results raise the possibility that if a seller could offer only a single variety in each period in the aforementioned papers’ settings, then there may be payoff multiplicity but the commitment payoff may remain achievable.\footnote{\citet{Kumar:06} studies such a setting and finds a unique equilibrium that does not yield the principal a high payoff. We attribute this to his model/analysis excluding the quality-price pair that would be used for leapfrogging. A similar point applies to \citet{inderst2008durable}, who studies a model with menus but finds that in some cases the principal's commitment payoff does not obtain.}

\paragraph{Renegotiation and Endogenous Status Quo:} Our model assumes that there is commitment to not renegotiate an accepted offer. A useful extension, which we do not pursue here, would be to model any agreement as the status quo for future negotiations; this would, of course, influence Vetoer's incentives to accept an offer insofar as it reveals information about her preferences that will affect future offers.
Although renegotiation has been studied in seller-buyer settings since \citet{hart1988contract} (see \cite{strulovici2017contract}, \citet{Maestri:17}, and \citet{GerardiMaestri:20} for recent contributions), the existing literature on political bargaining with an endogenous status quo, surveyed by \cite{EEZSurvey:20}, has generally not incorporated private information.

\section{Conclusion}\label{sec:conclusion}
Our paper has studied a canonical infinite-horizon model of sequential veto bargaining. We have shown how leapfrogging---making an offer that is accepted by some low types and rejected by some higher types---allows Proposer to alleviate his sequential rationality constraint and credibly extract surplus from high types; so much so that under some conditions, Proposer can (approximately) obtain his commitment payoff in an equilibrium when players are patient.

There are various directions that may be fruitful for future research. On the theoretical side, it would be of interest to incorporate ``pork'' or other forms of transfers in addition to the policy that our players have single-peaked preferences over. Studying a multidimensional policy is also important for political applications. On the empirical side, our work cautions against a presumption that Proposer's offers are successive concessions,\footnote{For example, in their survey article, \citet[p.~424]{CM:04} state a prediction that ``In sequential veto bargaining, Congress makes concessions in repassed bills'', as they did not consider the possibility of leapfrogging.} and calls for attention to whether and when we observe leapfrogging. Given that we have identified the coexistence of skimming and leapfrogging equilibria, norms in sequential veto bargaining with incomplete information are especially important; our results show how significantly Proposer could benefit from a favorable equilibrium. Laboratory experiments may be a fertile ground to deepen our understanding of equilibrium selection.

\vspace{.25in}

\begin{singlespace}
    \addcontentsline{toc}{section}{References}
    \bibliographystyle{ecta}
    \bibliography{AKK}
\end{singlespace} 

\appendix

\section{Proofs for Two-Type Example}

Recall that for the two-type example, we restrict attention to actions in $[0,1]$. The following proofs can be extended straightforwardly to handle actions outside $[0,1]$, but we omit that discussion for brevity.

\begin{lemma}
\label{lem:twotypes-skimming}
Fix any large $\delta<1$. Inductively define an increasing sequence $a^0:= 2l < a^1 < \ldots < a^N := 1$, where for each $i\geq 1$, $a^i$ is defined by either $u_V(a^{i},h)=\delta u_V(a^{i-1},h)$ if there is a solution with $a^i \in (a^{i-1},1]$, and otherwise $a^i:= 1$.\footnote{We suppress the dependence of $N$ and each $a^i$ (for $0<i<N$) on $\delta$.}
\begin{enumerate}[label={\normalfont (\alph*)}]
    \item \label{skimming1} If offers are restricted to lie in $[2l,1]$ then for any prior $\mu_0$ there is a skimming equilibrium in which, on path, Proposer first offers some $a^n$ with probability one and then works his way down the $(a^i)_{i=n}^0$ sequence to $2l$. Any offer $a^i>2l$ is rejected by type $l$ and accepted by type $h$ with positive probability. Both types accept the final offer of $2l$.
    \item \label{skimming2} Define $\mu^\delta \in (\mu^*,1)$ as the smallest belief that makes Proposer indifferent between the payoff from this (restricted) equilibrium and the payoff from leapfrogging, i.e., obtaining $a^\delta$ from type $l$ in the first period and action $1$ from type $h$ in the second period.\footnote{\label{fn:mudelta}The belief $\mu^\delta$ is well defined for large enough $\delta$. To confirm that, note first that for any $\mu_0 \leq \mu^*$, Proposer's payoff from leapfrogging, $\mu_0 \delta u(1)+(1-\mu_0)u(a^\delta)$ is strictly less than $u(2l)$ by definition of $\mu^*$ and that $a^\delta < 2l$; whereas his payoff from the (restricted) skimming equilibrium is at least $u(2l)$. Second, following the established seller-buyer analysis, for any interior belief $\mu_0$ Proposer's payoff in the (restricted) skimming equilibrium converges to $u(2l)$ as $\delta \to 1$, whereas leapfrogging's payoff converges to the strictly larger $\mu_0 u(1)+(1-\mu_0)u(2l)$. The result follows from continuity of both skimming and leapfrogging's payoffs in $\mu_0$.} If $\mu_0\leq \mu^\delta$, then the above skimming equilibrium exists without restriction on the space of offers: any offer in $(a^\delta,2l)$ is accepted by both types, while any offer in $[0,a^\delta]$ is accepted by $l$ and rejected by $h$. As $\delta \to 1$, $\mu^\delta \to \mu^*$.
    \item \label{skimming3} As $\delta \to 1$, Proposer's payoff in the above skimming equilibrium converges to $u(2l)$ regardless of his prior in the relevant range: for any $\epsilon>0$, there exists $\underline \delta<1$ such that if $\delta \in (\underline \delta,1)$ and $\mu_0 \leq \mu^\delta$, then Proposer's payoff in the skimming equilibrium is in $[u(2l),u(2l)+\epsilon)$.
\end{enumerate}
\end{lemma}

\begin{proof}
\underline{Part \ref{skimming1}}: Owing to the restriction to offers in $[2l,1]$, this part follows from arguments analogous to those in the
two-type seller-buyer bargaining problem (\citealp{Hart:86}; \citealp{FT:91}, pp.~{409--10}). So we omit a proof, instead only noting two points. First, if Proposer is indifferent between two first offers (as can also arise in the seller-buyer 
construction), we specify for concreteness that Proposer chooses the lower of the two. Second, there is one difference with the usual seller-buyer construction: if Proposer's first offer is $a^N=1$, and $a^N$ was defined by the action cap of $1$ rather than type $h$'s indifference, then Proposer will need to randomize on path between proposing $a^{N-1}$ and $a^{N-2}$ in the second round. Proposer's second-round randomization is chosen to make type $h$ indifferent between accepting and rejecting $a^N=1$; a suitable randomization exists because $h$ would strictly prefer accepting  $a^N=1$ if Proposer were to offer $a^{N-1}$ next, while $h$ would strictly prefer rejecting $a^N=1$ if Proposer were to offer $a^{N-2}$ next. Such on-path Proposer randomization is not necessary in the seller-buyer problem because there is no price cap---or, in effect equivalently, Proposer ideal point---there. 

\underline{Part \ref{skimming2}}: We stipulate that after a deviation in any period $t$ to $a_t<2l$, type $l$ accepts, whereas $h$ accepts if and only if $u_V(a_t,h) > \delta u_V(1,h)$, which is equivalent to $a_t> a^\delta$. After a rejection of the deviation, Proposer puts probability $1$ on type $h$ and proposes action $1$ ever after. Clearly we have an equilibrium in any continuation game after the initial deviation. So we need only verify that no deviation to $a_t<2l$ is profitable. Plainly, among $a_t\leq a^\delta$, the most profitable deviation is to $a^\delta$; but by definition of $\mu^\delta$, that deviation is not profitable when $\mu_t\leq \mu^\delta$. (A higher $\mu_t$ makes leapfrogging more attractive than the (putative) skimming equilibrium because Proposer prefers the skimming equilibrium when Vetoer is of type $l$ and leapfrogging when Vetoer is of type $h$.) Any deviation to $a_t \in (a^\delta,2l)$ yields a lower Proposer payoff than the (putative) skimming equilibrium because the skimming equilibrium's payoff is at least $u(2l)$. Therefore, no deviation to $a_t<2l$ is profitable when $\mu_t\leq \mu^\delta$, and the skimming equilibrium exists without any restriction on offers.

To see that $\mu^\delta \to \mu^*$ as $\delta \to 1$, observe that for any $\mu_0$, as $\delta \to 1$ Proposer's payoff from leapfrogging goes to $\mu_0 u(1)+(1-\mu_0)u(a^*)$ whereas, as discussed in \autoref{fn:mudelta},  his payoff from skimming goes to $u(2l)$. Hence, by definition of $\mu^*$, for any $\mu_0>\mu^*$, skimming is strictly worse than leapfrogging when $\delta$ is large enough. The result now follows from $\mu^\delta$ being the smallest belief at which the payoffs from skimming and leapfrogging are equal, noting that for any $\delta$ skimming yields a strictly higher payoff than leapfrogging at belief $\mu^*$ (see \autoref{fn:mudelta}).

\underline{Part \ref{skimming3}}: Given the previous two parts, this result follows from the same arguments as in the standard seller-buyer model \citep[e.g.,][pp.~{409--10}]{FT:91}.
\end{proof}

\begin{proof}[Proof of \autoref{prop:twotypes}]
Part \ref{prop:twotypes1} follows from \Cref{lem:twotypes-skimming}.

To prove parts \ref{prop:twotypes2} and \ref{prop:twotypes3}, we first define two critical values: $r^\delta(\mu)$ and the $\bar \mu^\delta$ referred to in the statement of the result. Recall $\mu^\delta\in (0,1)$ from \autoref{lem:twotypes-skimming}\ref{skimming2}. (In what follows, we sometimes suppress the caveat of ``for large $\delta$''.) 
For any belief $\mu\in (\mu^\delta,1)$, let \begin{equation}
\label{eq:rdelta}
 r^\delta(\mu) := \frac{\mu^\delta (1-\mu)}{(1-\mu^\delta)\mu}
\end{equation}
be type $h$'s rejection probability that would lead to posterior $\mu^\delta$ after rejection, given that type $l$ rejects with probability 1.
Now let $\bar \mu^\delta<1$ be the value of $\mu$ that solves\footnote{\label{fn:barmudelta}One can check that the difference between the LHS and the RHS of \autoref{eq:barmudelta} is continuous and strictly decreasing in $\mu$, strictly positive for small $\mu$, and strictly negative for large $\mu$; hence there is a unique solution, which is interior.}
\begin{equation}
(1-\mu)u(a^\delta) +\mu \delta u(1) =(1-\mu)\delta u(a^\delta) +\mu \left[1-r^\delta(\mu)+r^\delta(\mu)\delta^2\right]u(1).\label{eq:barmudelta}
\end{equation}
Given belief $\mu$, the LHS of \autoref{eq:barmudelta} is Proposer's utility from leapfrogging, whereas the RHS corresponds to  getting $a^\delta$ in the next period from $l$ and a lottery from $h$ of either action $1$ in the current period with probability $1-r^\delta$ or the same action in two periods with probability $r^\delta$. It can be verified that $\bar \mu^\delta > \mu^\delta$ and $\lim_{\delta \rightarrow 1} \bar \mu^{\delta}<1$.\footnote{\label{fn:mubardelta_limit}As $\mu \to \mu^\delta$ from above, $r^\delta(\mu)\to 1$, and so the RHS of \autoref{eq:barmudelta} goes to $\delta$ times the LHS, which is strictly smaller than the LHS. The properties noted in \autoref{fn:barmudelta} then imply $\bar \mu^\delta>\mu^\delta$. From \autoref{lem:twotypes-skimming}\ref{skimming2}, $\lim_{\delta \to 1} \mu^{\delta} = \mu^* \in (0,1)$. Algebraic manipulations of Equations \eqref{eq:rdelta} and \eqref{eq:barmudelta} yield $\lim_{\delta \rightarrow 1} \bar \mu^{\delta} \in (\mu^*,1)$.} 

\noindent\underline{Part \ref{prop:twotypes2}}: The equilibrium strategies, beliefs, and incentives are as follows.

\begin{enumerate}[wide, labelindent=0pt]
    \item Proposer proposes $a^\delta$ in the first period and $1$ in the second period (and ever after), with belief $\mu_t=1$ after any rejection. Vetoer type $l$ accepts in the first period while type $h$ rejects in the first period but accepts any proposal of at least $a^\delta$ starting in the second period. Clearly Proposer has no incentive to deviate starting in the second period, and Vetoer is playing optimally in all periods, so what we must show below is that Proposer has no incentive to deviate in the first period.
    
    \item \label{easydev0} (Region I in \autoref{Figure-Leapfrogging}.) If Proposer deviates and offers any action $a_0 \in [0,a^\delta)$ in the first period, type $l$ accepts and $h$ rejects. After a rejection, Proposer's belief is $\mu_t=1$ ever after and so he proposes $1$ ever after, which is accepted in the second period by type $h$. It is clear that Vetoer is playing optimally and that any such deviation is not profitable for Proposer.

    \item \label{easydev1} (Region II in \autoref{Figure-Leapfrogging}.) If Proposer deviates and offers any $a_0\in(a^\delta,2l]$ in the first period, both types accept that; for large $\delta$, this outcome is worse for Proposer than the on-path outcome, since the latter's payoff is larger than $u(2l)$. Both types accept any $a_0\in (a^\delta,2l]$ because we stipulate if any such offer is rejected (a zero probability event), Proposer holds belief $\mu_t=1$ ever after and offers action $1$ ever after.
    
    \item \label{easydev2} (Region III in \autoref{Figure-Leapfrogging}.) Let $u_h^*$ denote type $h$'s payoff in the skimming equilibrium discussed in \autoref{lem:twotypes-skimming} when Proposer has belief $\mu^\delta$ defined there. Since $\mu^\delta \to \mu^*$, it follows from the established seller-buyer analysis that for $\delta$ large enough, Proposer's first offer in our skimming equilibrium is arbitrarily close to $2l$ and hence $u_h^*$ is arbitrarily close to but strictly less than $u_V(2l,h)$. Let $\bar a^\delta>2l$ be such that $h$ is indifferent between accepting $\bar a^\delta$ in the current period and receiving payoff $u_h^*$ in the next period. 
    Note that $\bar a^\delta \approx 2l$ for large $\delta$.
    
Consider the interval $(2l,\bar a^\delta]$. As described in \autoref{lem:twotypes-skimming}, the skimming equilibrium (defined assuming actions constrained in $[2l,1]$) is constructed using a sequence of actions $a^{0}\equiv 2l<a^{1}<\ldots<a^{N}\equiv 1$
that is defined by $h$'s indifference. (We suppress the dependence of the sequence on $\delta$ to reduce notation.) Let $M\leq N-1$ be such that $a^M < \bar a^\delta\leq a^{M+1}$.

For any deviation $a_0\in (2l,a^1]$, $l$ rejects and $h$ accepts; Proposer holds belief $\mu_t=0$ and offers $a_t=2l$ ever after (accepted by type $l$ in the second period).

Suppose $\bar a^\delta>a^1$. For any deviation $a_0\in (a^{1},\bar a^\delta]$, let $n\in\{1,\ldots,M\}$ be such that $a_0\in(a^{n},a^{n+1 }]$.
Type $l$ rejects, while type $h$ rejects with the probability that makes the posterior $\mu_1=\mu^n$, where $\mu^n$ is the unique belief that makes Proposer indifferent between starting the decreasing offer sequence with $a^{n}$ and $a^{n-1}$. (Type $h$'s rejection probability is well-defined and unique so long as $\mu^n\leq \mu_0$, which will be verified below by showing that $\mu^n\leq \mu^\delta$.) Proposer will then randomize in the second period between the starting offers of $a^{n}$ and $a^{n-1}$. If Proposer were to start with $a^{n}$, $h$ would prefer to accept $a_0$; if Proposer were to start with $a^{n-1}$, $h$ would prefer to reject $a_0$; so there is a unique randomization that makes $h$ indifferent.
We are left to check that $\mu^n\leq \mu^\delta$: if so, then Proposer prefers the decreasing offer sequence to leapfrogging, and we can support the skimming equilibrium by specifying behavior for offers in $[0,2l]$ as in the proof of \autoref{lem:twotypes-skimming}\ref{skimming2}. Indeed $\mu^n\leq \mu^\delta$, since $n\leq M$ and under belief $\mu^\delta$ Proposer starts the decreasing offer sequence with $a^{M}$ while under belief $\mu^n$ it is optimal to start with $a^{n}$ (and a higher belief corresponds to a higher starting offer in the skimming equilibrium).\footnote{That Proposer starts the decreasing offer sequence with $a^{M}$ under belief $\mu^\delta$ follows from type $h$'s indifference in the definition of $\bar a^\delta$ and $\bar a^\delta \in (a^M,a^{M+1}]$. For, if Proposer started with an offer $a^{M-1}$ or lower, then $h$ would strictly prefer to wait for that offer in the next period rather than accept $\bar a^\delta$ in the current period; if Proposer started with an offer $a^{M+1}$ or higher, then $h$ would strictly prefer to accept $\bar a^\delta$.}

So a deviation to any $a_0 \in (2l,\bar a^\delta]$ yields Proposer a payoff that is no higher than from a skimming equilibrium with restricted action space $[2l,1]$ and belief $\mu_0$ (see \autoref{lem:twotypes-skimming}\ref{skimming1}). As $\delta \to 1$, the payoff from a (restricted) skimming equilibrium converges uniformly to $u(2l)$ on any interval of priors bounded away from $1$, whereas the payoff from leapfrogging converges uniformly to $\mu_0 u(1)+(1-\mu_0) u(a^*)$. The latter limit is strictly larger than the former limit when $\mu_0>\mu^*$, by definition of $\mu^*$. Since $\mu^\delta>\mu^*$ and $\lim_{\delta \to 1} \bar \mu^\delta <1$, it follows that for all $\delta$ large enough, the payoff from leapfrogging is strictly larger than from the (restricted) skimming equilibrium for all $\mu_0 \in (\mu^\delta,\bar \mu^\delta)$. Hence, for $\delta$ large enough, a deviation to any $a_0 \in (2l,\bar a^\delta)$ is not profitable.

    \item \label{interestingdev} (Region IV in \autoref{Figure-Leapfrogging}.) It remains to consider any first-period deviation $a_0\in (\bar a^\delta,1]$.
    \begin{itemize}
        \item Type $l$ rejects since $a_0>2l$. Type $h$ rejects with probability $r^\delta(\mu_0)$, independent of $a_0$, which leads to second-period belief $\mu_1=\mu^\delta$.

        \item In the second period: Proposer randomizes between starting the play of a skimming equilibrium (see \autoref{lem:twotypes-skimming}) with some probability $\lambda(a_0)$ and starting the leapfrogging path with remaining probability. By definition of $\mu^\delta$, Proposer is indifferent between starting either of these two paths. The randomization probability $\lambda(a_0)$ is set to make type $h$ indifferent between accepting $a_0$ in the first period and getting a lottery over payoff $u^*_h$ in the second period with probability $\lambda(a_0)$ and getting action $1$ in the third period with complementary probability.\footnote{I.e.,
        $u_V(a_0,h)=\lambda(a_0) \delta u_h^*+(1-\lambda(a_0))\delta^2 u_V(1,h)$. There is a unique $\lambda(a_0)$ that solves this equation because $\delta u^*_h>u_V(a_0,h)>\delta^2 u_V(1,h)$, where the first inequality is because $a_0>\bar a^\delta>h$ and $\delta u^*_h=u_V(\bar a_\delta,h)$.}
        
        For any second-period offer $a_1$ besides the two that Proposer randomizes over, we stipulate that continuation play would follow that in a skimming equilibrium with initial offer $a_1$. Plainly, no such offer $a_1$ is a profitable deviation.

        \item Finally, we argue that among deviations to $a_0\in (\bar a^\delta,1]$, the most profitable deviation is to action $1$, and that is not profitable because $\mu_0\leq \bar \mu^\delta$. Note that after a rejection of any $a_0>\bar a^\delta$, leapfrogging is optimal for Proposer in the second period. So Proposer's expected payoff from any $a_0 > \bar a^\delta$ is
$$(1-\mu_0)\delta u(a^\delta)+\mu_0\left[\left(1-r^\delta(\mu_0)\right)u(a_0)+r^\delta(\mu_0)\delta^2 u(1)\right].$$
This payoff is maximized when $a_0=1$, in which case it becomes the RHS of \autoref{eq:barmudelta} (with $\mu=\mu_0$). Since $\mu_0\leq\bar \mu^\delta$, the definition of $\bar \mu^\delta$ implies that leapfrogging starting in the first period is at least as good for Proposer (see \autoref{fn:barmudelta}).

    \end{itemize} 
\end{enumerate}

\noindent\underline{Part \ref{prop:twotypes3}}: The construction for this part is the same as that for part \ref{prop:twotypes2}, except that Proposer now proposes action 1 in the first period, rather than $a^\delta$. By the logic used in the last bullet of point \ref{interestingdev} above, proposing $a_0=1$ is better for Proposer than proposing any $a_0\in (\bar a^\delta,1)$, and also now better than proposing $a_0=a^\delta$ because $\mu_0>\bar \mu^\delta$. By points \ref{easydev0}--\ref{easydev2} above, $a_0=a^\delta$ is in turn better than any other first-period offer less than $\bar a^\delta$.
\end{proof}

\section{Proofs for General Analysis}

\subsection{Obtaining the Commitment Payoff without Commitment} 
\begin{proof}[Proof of \autoref{lemma:static_provides_upper_bound}]
Fix any strategy for Proposer and any best response for Vetoer, and denote this strategy profile by $\sigma$. For any type $v$, the profile $\sigma$ induces a probability distribution $\lambda_v$ over $\mathbb R \times \mathbb{N}\cup \{\infty\}$, where $(a,t)\in \mathbb R \times \mathbb N$ denotes the outcome that proposal $a$ is accepted in period $t$, and $\infty$ denotes no agreement. We construct an incentive compatible and individually rational mechanism for the static problem that achieves the same expected payoff for Proposer as under $\sigma$. 

For any $t\in\mathbb{N}$, let $\lambda_v(t)$ be the measure on $\Reals$ defined by $\lambda_v(t)(A):= \lambda_v(A\times\{t\})$ for every (Borel) set $A\subseteq \mathbb R$. Define a mechanism for the static problem as follows:
\[ m(v):= \sum_{t=0}^{\infty} \delta^t \lambda_v(t)+ \Big( 1- \sum_{t=0}^{\infty} \delta^t \lambda_v(t)(\Reals) \Big) \mathbf{1}_0,\]
where $\mathbf{1}_0$ denotes the Dirac measure on $0$. Intuitively, for every agreement $(a,t)$ that has positive probability under $\lambda_v$, $m(v)$ gives probability $\delta^t$ to action $a$ and probability $1-\delta^t$ to action $0$. It can be verified that $m(v)$ is a probability measure over $\Reals$. 

Since
\[ \int_{a} u_V(a,v)  \mathrm d m(v)(a)= \sum_{t=0}^{\infty} \delta^t \int_{a} u_V(a,v)  \mathrm d \lambda_{v'}(t)(a), \]
the expected utility for type $v$ reporting $v'$ in the static mechanism is the same as in the dynamic game were type $v$ to play as $v'$ does.  Hence, as  
Vetoer is playing a best response in $\sigma$, mechanism $m$ is incentive compatible and individually rational.

Analogous arguments show that Proposer's expected utility in the static mechanism is the same as his expected utility in the dynamic game under strategy profile $\sigma$. Therefore, Proposer can replicate his payoff from the dynamic game using a static mechanism, and hence can do no worse in the static problem.
\end{proof}

\begin{proof}[Proof of \autoref{lemma:opt_delegation_set}]
To obtain a contradiction, suppose there is a (potentially stochastic) mechanism $m$ that yields a strictly higher payoff than the delegation set $[c^*,1]$ under prior $F_{[c,c']}$ for some $c\le c^*/2\le c^*\le c'$. Let $M:=m([c,c'])$ denote the image of $[c,c']$ under $m$. We can assume without loss of generality that $u(m(c'))\ge u(m(v))$ for all $v\in[c,c']$ and that $u(m(c))\ge u(0)$.\footnote{If $u(m(c'))< u(m(v))$ for some $v\in[c,c']$, add the action $\min\{1,\mathbb{E}_{m(c')}[a]\}$ to $M$ and consider the corresponding mechanism $\hat m$ in which each type chooses its favorite lottery, breaking indifference in Proposer's favor. Since $\mathbb{E}_{m(v)}[a]$ is increasing in $v$ 
because mechanism $m$ is IC, the new mechanism $\hat m$ yields Proposer a higher payoff than $m$ and satisfies $u(\hat m(c'))\ge u(\hat m(v))$ for all $v\leq c'$.

Now suppose $u(m(c))<u(0)$. If $c\leq 0$, consider an alternative mechanism $\hat m$ that is identical to $m$ except for assigning action $0$ with probability one to all types below $0$. This mechanism is individually rational (IR) and IC and yields Proposer a higher payoff than $m$ and satisfies $u(\hat m(c))=u(0)$. If $c>0$, consider an alternative mechanism $\hat m$ that is identical to $m$ except for $\hat m(c)$ assigning probability one to an action in $[0,\E[m(c)]$ that makes type $c$ indifferent with $m(c)$. Such an action exists because $u_V(m(c),c)\geq u_V(0,c)$, as $m$ is IR, and $u_V(\cdot,c)$ is continuous. 
Since $m$ is IC and IR, and any type $v > c$ prefers $m(v)$ to $\hat m(c)$ (by SCED, type $c$'s indifference between $m(c)$ and $\hat m(c)$, and that type $\hat m(c)$ strictly prefers $\hat m(c)$ to $m(c)$), it follows that $\hat m$ is incentive compatible and individually rational. Moreover, $u(\hat m(c))\ge u(0)$ and the mechanism $\hat m$ yields Proposer a higher payoff than $m$.}   Define a menu of stochastic actions by
\begin{align*}
\tilde{M} := M\cup \{v\in[c^*,1]: u(v)\ge u(m(c'))\}\cup\{0\}.
\end{align*}
Let $\tilde{m}$ be the induced mechanism where each type $v$ chooses its favorite action in $\tilde{M}$ and indifference is broken in Proposer's favor.  Plainly, $\tilde m$ is incentive compatible and individually rational.
We will show that given prior $F$, Proposer's payoff from $\tilde{m}$ is strictly higher than from  delegation set $[c^*,1]$. 

Conditional on the event $\{v:v\in [c,c']\}$, Proposer's payoff from menu $M$ is strictly higher than from menu $[c^*,1]$ by assumption. Compared to menu $M$, the additional actions in $\tilde{M}$ chosen by types $v\in[c,c']$ are ones that Proposer prefers to $m(c')$, which he prefers to $m(v)$ for any $v\in[c,c']$. Hence,  conditional on $\{v:v\in [c,c']\}$, Proposer's payoff from menu $\tilde{M}$ is strictly higher than from menu $[c^*,1]$.

We next show that for every $v> c'$, $u(\tilde{m}(v))\ge u(v)$. Since Vetoer's utility satisfies SCED and she breaks indifference in favor of Proposer, either $\tilde{m}(v)=m(c')$ or $\tilde{m}(v)\in \tilde{M}\setminus (M\cup \{0\})$. In either case, $u(\tilde{m}(v))\ge u(m(c'))$. If $u(m(c'))>u(v)$ then it follows that $u(\tilde{m}(v))\ge u(v)$. If, instead, $u(v)\ge u(m(c'))$ then $\tilde{m}(v)=v$ and we conclude $u(\tilde{m}(v))= u(v)$.

Moreover, SCED implies that for all $v<c$, either $\tilde{m}(v)=\tilde m(c)$ or $\tilde m(v)=0$. Since $u(\tilde{m}(c))\ge u(0)$ and $u(0)$ is Proposer's payoff under delegation set $[c^*,1]$ whenever $v<c$, it follows that Proposer's payoff from mechanism $\tilde{m}$ is  higher than his payoff from delegation set $[c^*,1]$ under belief $F$, a contradiction.
\end{proof}

\begin{proof}[Proof of \autoref{lemma:induction_base}]
Fix any $\varepsilon>0$. Let $\underline{\delta}<1$ be such that $\underline{\delta} U(F_{[\underline{v},c^*]})\ge U(F_{[\underline{v},c^*]})-\varepsilon$, and fix any $\delta\ge \underline{\delta}$.
Let $(\tilde{\sigma},\tilde{\mu})$ be an equilibrium when Proposer's prior belief is $F_{[\underline{v},c^*]}$, where $\tilde{\sigma}$ denotes the strategy profile and $\tilde{\mu}$ the system of beliefs. If Proposer's payoff in equilibrium $(\tilde{\sigma},\tilde{\mu})$ is higher than $\delta U(F_{[\underline{v},c^*]})$ then the claim holds; so suppose his payoff is strictly lower.
Define a candidate equilibrium profile $(\sigma,\mu)$ as follows: 
\begin{itemize}
\item On path, Proposer offers $0$ in the first period, $c^*$ in the second period, followed by $\min\{2c^*,1\}$ ever after. Vetoer of type $v$ accepts the first proposal $0$ if and only if she strictly prefers it to $c^*$ in the next period;
in the second period she accepts $c^*$ if and only if she (weakly) prefers it to both $\min\{2c^*,1\}$ and $0$ in the third period; and for any subsequent history starting with proposal sequence $(0,c^*)$, she accepts the current proposal if and only if she (weakly) prefers it to both $\min\{2c^*,1\}$ and $0$ in the next period. For any on-path history $h$, let $\mu(h)$ be derived from Bayes' rule whenever possible, and for any history $h$ starting with $(0,c^*)$, let $\mu(h)$ put probability 1 on type $c^*$.
\item For any off-path history $h$ that starts with $(0,a)$ for $a\neq c^*$, let $(\sigma,\mu)$ specify some continuation equilibrium with the starting belief $F_{[c^*/2,c^*]}$; a continuation equilibrium exists by hypothesis \eqref{eqmexists}. 
For any off-path history $h$ in which the first proposal is different from 0, let $(\sigma,\mu)(h)=(\tilde{\sigma},\tilde{\mu})(h)$.
\end{itemize}

Proposer's payoff from the strategy profile $\sigma$ is $\delta U(F_{[\underline{v},c^*]})$ because on path types below $c^*/2$ accept proposal 0 and types in $[c^*/2,c^*]$ accept proposal $c^*$ in period 1; while in the static problem, \autoref{lemma:opt_delegation_set} implies that for belief $F_{[\underline{v},c^*]}$ the delegation set $[c^*,1]$ is optimal, which results in all types in $[\underline{v},c^*/2)$ obtaining action $0$ and all types in $(c^*/2,c^*]$ obtaining action $c^*$. We will argue that the profile $(\sigma,\mu)$ is an equilibrium, which proves the claim.

First, Proposer is playing a best response in the profile $(\sigma,\mu)$ at the start of the game since any deviation induces the same payoff as in equilibrium $(\tilde{\sigma},\tilde{\mu})$, which is strictly lower than $\delta U(F_{[\underline{v},c^*]})$ by hypothesis. Moreover, by construction, Vetoer is playing a best response at the history $h=(0)$, i.e., after the initial proposal of $0$. 

Second, we claim that Proposer is playing a best response at history $h=(0)$. Note that the second-period belief after this history is $\mu(0)=F_{[c^*/2,c^*]}$ and that in the continuation game starting at $h=(0)$ the strategy profile $\sigma$ yields payoff $U(F_{[c^*/2,c^*]})$: all types in $[c^*/2,c^*]$ accept proposal $c^*$ immediately and the delegation set $[c^*,1]$ solves the static problem by \autoref{lemma:opt_delegation_set}.
Any deviation by Proposer to an offer $a\neq c^*$ gives Proposer a payoff 
of at most $U(F_{[c^*/2,c^*]})$ by \autoref{lemma:static_provides_upper_bound}.
Therefore, Proposer is playing a best response at history $h=(0)$. 

Finally, we claim that both players are playing best responses at any other history. Indeed, for any history starting with proposals $(0,c^*)$, best responses are assured by construction. For any history starting with $(0,a)$ with $a\neq c^*$, our construction specifies some continuation equilibrium. For any history starting with a proposal different from $0$ players are playing an equilibrium because $(\tilde{\sigma},\tilde{\mu})$ is an equilibrium for prior belief $F_{[\underline{v},c^*]}$.

As it is straightforward that the system of beliefs $\mu$ satisfies Bayes Rule whenever possible, we conclude that $(\sigma,\mu)$ is an equilibrium.
\end{proof}

\begin{proof}[Proof of \autoref{prop:commitment_achievable_general}]

Without loss of generality, we assume $U(F)\leq 1$, as Proposer's utility can be rescaled accordingly. Furthermore, we prove the result only for $c^*>0$; the $c^*=0$ case is implied by \autoref{prop:existence_decreasing_offers}. 

As a roadmap: Steps 1--4 below use induction to show that there are equilibria in which Proposer can obtain arbitrarily close to his commitment payoff on some interval of types below a threshold. Step 5 establishes this threshold can be made arbitrarily close to $\ov$. Step 6 then argues that there is an equilibrium in which Proposer obtains arbitrarily close to his commitment payoff from the full interval of types $[\uv,\ov]$.

We begin with some preliminaries for the inductive argument. Let $c_0(\varepsilon, \delta):=c^*>0$ and define for all integers $n>0$, $$c_n(\varepsilon, \delta):=\min\left\{c_{n-1}(\varepsilon,\delta)+\frac{\varepsilon}{4u'(0)}, c_{n-1}(\varepsilon,\delta) \sqrt{1+\sqrt{1-\delta}} \right\}.\footnote{If $u$ is not differentiable at $0$, let $u'(0)$ denote the right-derivative at $0$, which exists because $u$ is concave.}$$ It follows that there is some $n\in \mathbb{N}$ such that $c_n(\varepsilon,\delta)\ge \overline{v}$. Let $\underline{f}>0$ denote a lower bound for $f$ on $[\underline{v},\overline{v}]$.
For $\epsilon>0$, define $$\delta^*(\epsilon,\delta) := 1- \frac{\epsilon}{2} \underline{f}\min\left\{\frac{\epsilon}{4u'(0)}, c^* \left(\sqrt{1+\sqrt{1-\delta}}-1\right)\right\},$$ and let $\underline{\delta}(\epsilon)\in(\sqrt{1-\epsilon},1)$ be such that for all $\delta\in(\underline{\delta}(\epsilon),1)$, $\delta\ge \delta^*(\epsilon,\delta)$. Such a $\underline{\delta}(\epsilon)$ exists because $\delta^*(\epsilon,1)=1$, $\delta^*(\epsilon,\cdot)$ is continuous, and $\lim_{\delta \uparrow 1} \frac{\partial \delta^*(\epsilon,\delta)}{\partial \delta}=+\infty$.

The induction hypothesis for $n\geq 0$ is:

\smallskip

For all $\varepsilon>0$, $\delta>\underline{\delta}(\epsilon)$, and $c$ satisfying $c^*\le c\le c_n(\varepsilon,\delta)$, if Proposer's belief is $F_{[\underline{v},c]}$ then  there is an equilibrium in which Proposer's payoff is at least $U(F_{[\underline{v},c]})-\epsilon$. 

\smallskip

The induction hypothesis holds for $n=0$ by \autoref{lemma:induction_base}.

Let $(\hat{\sigma},\hat{\mu})$ be an equilibrium for the game with belief $F_{[\underline{v},c_{n-1}\ed]}$ that yields Proposer payoff at least $U(F_{[\underline{v},c_{n-1}\ed]})-\epsilon$ (such an equilibrium exists under the induction hypothesis) and let $a_{n-1}\ed$ be the largest action that makes type $c_{n-1}(\epsilon,\delta)$ indifferent between accepting $a_{n-1}\ed$ and playing $(\hat{\sigma},\hat{\mu})$ from the next period on.
Steps 1--4 below establish that if the induction hypothesis holds for $n$ and $a_{n-1}\ed \le 1$ then it holds for $n+1$, given \eqref{eqmexists}.

\underline{Step 1}:  Fix arbitrary $\varepsilon>0$, $\delta\ge \underline{\delta}(\epsilon)$, and $c$ satisfying $c_n(\varepsilon,\delta)< c\le c_{n+1}(\varepsilon,\delta)$, and an equilibrium $(\tilde{\sigma},\tilde{\mu})$ for the game with belief $F_{[\underline{v},c]}$. If Proposer's payoff is at least $U(F_{[\underline{v},c]})- \epsilon$ we are done; so suppose Proposer's payoff is strictly less. Below, we suppress the dependence of $c_n$ and $a_{n-1}$ on $\varepsilon$ and $\delta$, and we set $c_{-1}(\epsilon,\delta):=c^*$.

 We construct a new equilibrium $(\sigma,\mu)$ for the game with belief $F_{[\underline{v},c]}$ as follows:
Proposer's first offer is $a_{n-1}$. On path, types above $c_{n-1}$ accept $a_{n-1}$ and types below $c_{n-1}$ reject $a_{n-1}$. After a rejection of $a_{n-1}$, Proposer updates to $F_{[\underline{v},c_{n-1}]}$ and continuation play proceeds as specified by $(\hat{\sigma},\hat{\mu})$. 
Moreover, if Proposer deviates in the first period, continuation play is as specified by $(\tilde{\sigma},\tilde{\mu})$.

\underline{Step 2}: We show that Vetoer is playing a best response when $a_{n-1}$ is proposed in the first period. 

It is optimal for types below $c_{n-1}$ to reject $a_{n-1}$ since type $c_{n-1}$'s equilibrium strategy in the continuation game yields a higher payoff (using that $a_{n-1}> c_{n-1}$ and Vetoer's preferences satisfy SCED).\footnote{To elaborate, note that when comparing action $a_{n-1}$ and the lottery induced by type $c_{n-1}$'s equilibrium strategy, $c_{n-1}$ is indifferent whereas (a possibly hypothetical) 
type $a_{n-1}$ strictly prefers action $a_{n-1}$. SCED implies that given any two lotteries and any three types $v_1<v_2<v_3$, if $v_2$ is indifferent and $v_3$ strictly prefers one lottery, then $v_1$ (weakly) prefers the other lottery.} We now explain why it is optimal for types in $[c_{n-1},c]$ to accept $a_{n-1}$; there is no need to consider types above $c$ because Proposer's belief is supported on $[\uv,c]$. Accepting $a_{n-1}$ is a best response for types $c_{n-1}$ and $a_{n-1}$, and SCED implies that the set of types for which it is a best response to accept is an interval. Therefore, if $a_{n-1}\ge c$, then  accepting $a_{n-1}$ is a best response for all types in $[c_{n-1},c]$.
So suppose $a_{n-1}\in [c_{n-1},c)$. It would be a best response for type $c$ to accept $c_{n-1}$ since that is even better than obtaining $c$ next period (as $2 c c_{n-1} -c_{n-1}^2 \ge \delta c^2$ because of our assumption that $c\le c_{n-1} + c_{n-1}\sqrt{1- \delta }$). Therefore, since type $c$ prefers $a_{n-1}\in [c_{n-1},c)$ to $c_{n-1}$, accepting $a_{n-1}$ is a best response for type $c$ and hence for all types in $[c_{n-1},c]$.

\underline{Step 3}: 
We show that Proposer's payoff from profile $\sigma$ is at least $U(F_{[\underline{v},c]}) - \epsilon$.

Proposer's payoff if the first proposal $a_{n-1}$ is accepted times the probability of acceptance is at least
\begin{align*}  
[F_{[\underline{v},c]}(c)-F_{[\underline{v},c]}(c_{n-1})] u(c_{n-1})&\ge
\int_{c_{n-1}}^{c} [u(v)-u'(0)(v-c_{n-1})]\mathrm dF_{[\underline{v},c]}\\
&\ge
\int_{c_{n-1}}^{c} [u(v) - \varepsilon/2] \mathrm dF_{[\underline{v},c]}, 
\end{align*}
where the first expression is because $a_{n-1}\in[c_{n-1},1]$,
the first inequality is because $u(v)-u(c_{n-1})\le u'(0)(v-c_{n-1})$,
and the second inequality is because $c-c_{n-1}\le \frac{\varepsilon}{2u'(0)}$.

For the case $n=0$, Proposer's payoff conditional on proposal $a_0$ being rejected times the probability of rejection is at least $\delta^2 U(F_{[\underline{v},c_{n-1}]})F_{[\underline{v},c]}(c_{n-1})$ by \autoref{lemma:induction_base}. Since $\delta\ge \sqrt{1-\epsilon}$ and $U(F_{[\underline{v},c_{n-1}]})\le 1$, these two bounds imply that Proposer's payoff is at least
\[[U(F_{[\uv, c_{n-1}]})-\epsilon]F_{[\uv, c]}(c_{n-1}) + \int_{c_{n-1}}^{c} [u(v) - \varepsilon/2] \mathrm dF_{[\underline{v},c]}. \]
Since the delegation set $[c^*,1]$ is optimal for belief $F_{[\uv,c]}$ by \autoref{lemma:opt_delegation_set}, this implies that Proposer's payoff is at least $U(F_{[\underline{v},c]}) - \epsilon$.

Consider now the case $n \geq 1$. Proposer's payoff conditional on proposal $a_{n-1}$ being rejected times the probability of rejection is at least
$\delta \left[U(F_{[\underline{v},c_{n-1}]})-  \epsilon\right] F_{[\underline{v},c]}(c_{n-1})$.  
Therefore, Proposer's payoff is at least 
\begin{align*}
&    \delta \left[U(F_{[\underline{v},c_{n-1}]})-  \epsilon\right] F_{[\underline{v},c]}(c_{n-1})+ \int_{c_{n-1}}^{c} [u(v) - \varepsilon/2] \mathrm dF_{[\underline{v},c]}\\
\geq & U(F_{[\underline{v},c]})-  \epsilon+  \frac{\varepsilon}{2} [F_{[\underline{v},c]}(c)-F_{[\underline{v},c]}(c_{n-1})] - (1-\delta)\\
\geq &  U(F_{[\underline{v},c]})- \epsilon,
\end{align*}
where the first inequality is because the delegation set $[c^*,1]$ is optimal for belief $F_{[\underline{v},c]}$ (by \autoref{lemma:opt_delegation_set}) and $U(F_{[\underline{v},c_{n-1}]})\le 1$,
and the second inequality is because
\[    F_{[\underline{v},c]}(c)-F_{[\underline{v},c]}(c_{n-1})\ge \underline{f} \min\left\{\frac{\varepsilon}{4u'(0)},c^*\left(\sqrt{1+\sqrt{1-\delta}}-1\right)\right\} \] 
and
\[ \delta\ge \delta^*(\epsilon,\delta)= 1- \frac{\varepsilon}{2} \underline{f}\min\left\{\frac{\varepsilon}{4u'(0)},c^*\left(\sqrt{1+\sqrt{1-\delta}}-1\right)\right\}.\]
This establishes Step 3.

\underline{Step 4}: To verify that $(\sigma,\mu)$ is an equilibrium, observe that
Proposer plays a best response in the first period since any deviation gives a payoff less than $U(F_{[\underline{v},c]})-\epsilon$ by supposition. Vetoer plays a best response to proposal $a_{n-1}$ as argued above. Finally, both players play best responses after any other history because we began in Step 1 with equilibria 
$(\tilde{\sigma},\tilde{\mu})$ and  $(\hat{\sigma},\hat{\mu})$. This establishes the induction step if $a_{n-1}\le 1$.

\underline{Step 5}: We show that, when $\epsilon$ is small and $\delta$ is large, the inductive argument in Steps 1--4 covers a large fraction of types.

Let $\bar c(\epsilon,\delta):=c_{n}(\epsilon,\delta)$, where $n$ is the smallest index such that the action $a_{n}\ed$ defined in our induction argument is strictly above 1.
We claim
\begin{align}\label{eq:c_n-1-bound}
\int_{\bar c(\epsilon,\delta)}^{\ov} \left[u(1)-u(v)\right] \mathrm{d} F(v) \le 1-\delta+\varepsilon,
\end{align}
which implies that $\bar c(\epsilon,\delta) \to \ov$ as $\delta\rightarrow 1$ and $\epsilon\rightarrow 0$.

To derive inequality \eqref{eq:c_n-1-bound}, note that because $a_{n}\ed > 1$ there is a static incentive compatible and individually rational mechanism in which all types above $\bar{c}\ed$ receive action $1$ and Proposer's payoff from types below $\bar{c}\ed$ is at least as in equilibrium $(\hat{\sigma},\hat{\mu})$   discounted by $\delta$. This mechanism gives Proposer payoff at least
 \[\delta [U(F_{[\underline{v},\bar{c}\ed]})-\varepsilon]F(\bar{c}\ed) +\int_{\bar{c}\ed}^{\ov}u(1) \mathrm d F(v).\] 
 By \autoref{lemma:opt_delegation_set}, this is less than the payoff from delegation set $[c^*,1]$, which can be written as
 \[ U(F_{[\uv,\bar{c}\ed]})F(\bar{c}\ed)+\int_{\bar{c}\ed}^{\ov}u(v) \mathrm d F(v).\]
 Some algebra using $U(F_{[\uv,\bar{c}\ed]})\le 1$ now yields inequality \eqref{eq:c_n-1-bound}.
 
 \underline{Step 6}: Given the belief $F$ and an arbitrary $\epsilon>0$, we show that for all $\delta$ large enough there is an equilibrium in which Proposer's payoff is at least $U(F)-\epsilon$, which completes the proof.
 
 For any $\epsilon'>0$ and $\delta> \underline{\delta}(\epsilon')$, we have established in Steps 1--5 that for belief $F_{[\uv,\bar c(\epsilon',\delta)]}$  there is an equilibrium, denoted by $(\sigma, \mu)$, in which Proposer's payoff is at least $U(F_{[\uv,\bar c(\epsilon',\delta)]})-\epsilon'$. Let $a(\epsilon',\delta)$ be the largest action that makes type $\bar c(\epsilon',\delta)$ indifferent between accepting $a(\epsilon',\delta)$ and playing $(\sigma,\mu)$ from next period on. Note that $a(\epsilon',\delta) \in (1,2]$ by definition (that it is less than $2$ is because actions above 2 are worse than the status quo for all types).
 
 Consider a strategy profile in which Proposer initially offers $a(\epsilon',\delta)$, followed by continuation play as described by $(\sigma,\mu)$. It is a best response for all types in $[\bar c(\epsilon',\delta),\ov]$ to accept $a(\epsilon',\delta)$ because of SCED and that accepting is a best response for type $\bar c(\epsilon',\delta)$ and a (hypothetical) type $a(\epsilon',\delta)$ that is larger than $\ov$; it is also a best response for all types below $\bar c(\epsilon',\delta)$ to reject $a(\epsilon',\delta)$.
 Since $a(\epsilon',\delta)\in (1,2]$ and Proposer's ideal point is $1$, it follows that Proposer's payoff given this strategy profile is at least
 \begin{align*}
     \delta [U(F_{[\uv,\bar c(\epsilon',\delta)]})-\epsilon'] F(\bar c(\epsilon',\delta)) + \int_{\bar c(\epsilon',\delta)}^{\ov} u(2) \mathrm dF(v).
 \end{align*}
 For $\epsilon'>0$ small enough and $\delta<1$ large enough, this payoff is at least $U(F)-\epsilon$. Given \eqref{eqmexists} it follows that there is an equilibrium in which Proposer's payoff is at least as large: analogous to the logic used in Step 1, if a given equilibrium does not yield payoff at least $U(F)-\epsilon$, we can modify it by having Proposer offer $a(\epsilon',\delta)$ in the first period with continuation play given by $(\sigma,\mu)$.
\end{proof}

\subsection{A Skimming Equilibrium}
\label{app:decreasing}
We construct a skimming equilibrium building on ideas from the seller-buyer literature, which are summarized instructively by \citet[pp.~1912--15]{ACD:02}. Our first step is to define a pair of functional equations whose joint solution describes a skimming equilibrium.

\begin{definition}\label{def:support_eq}
Let $R:[\uv,v^*]\to \Reals$ be continuous and $P:[\uv,v^*]\to \Reals$ be right-continuous, where  $v^*\in (\uv,\ov]$. We say that $(R,P)$ \emph{supports a skimming equilibrium on $[\underline{v},v^*]$} if, for all $v\in [\uv,v^*]$,
\begin{align}
    \label{eq:Bellman} R(v)&=\max_{y\in[\uv, v]}\left\{u(\oP(y))[F(v)-F(y)]+\delta R(y)\right\},\\
    \label{eq:Bellman2} u_V(P(v),v) &=\delta u_V(\oP(t(v)),v),
\end{align}
where $T(v)$ denotes the argmax correspondence in  \eqref{eq:Bellman}, $t(v):=\max T(v)$,  $P(v)$ is the largest proposal that satisfies \eqref{eq:Bellman2}, and $\oP$ is the increasing envelope of $P$, i.e., $\oP(v):= \sup_{y\le v}P(y)$.\footnote{The maximizers in this definition exist because $P$ being right-continuous implies $\oP$ is right-continuous, and since it is also increasing, $\oP$ is upper semicontinuous.} 
\end{definition}

The idea behind this definition is that $R(y)$ describes Proposer's value function and $\oP(v)$ describes Vetoer's acceptance behavior. We will construct an equilibrium in which at any history, type $v$ accepts a positive offer if and only if the offer is below $\oP(v)$. Alternatively, given that $\oP$ is increasing, any offer $\oP(v)$ is accepted precisely by all types above $v$.\footnote{This statement is imprecise when there are multiple $\tilde v$ such that $\oP(\tilde v)=\oP(v)$; we gloss over this issue for this heuristic explanation.} Consequently, at any history, Proposer's belief is a right-truncation of the prior to $[\uv,v]$ for some $v$. The upper endpoint $v$ thus acts like a state variable that Proposer optimizes. \Cref{eq:Bellman} is the dynamic programming equation that captures Proposer's tradeoff between extracting surplus via screening and the cost of delay: given the current state $v$, if Proposer brings the state down to $y$ with an offer $\oP(y)$, then with probability $F(v)-F(y)$ (ignoring a normalization factor) he obtains current payoff $u(\oP(y))$; in addition, after a one-period delay he obtains payoff $R(y)$. 
Concomitantly, \Cref{eq:Bellman2} is the indifference condition for type $v$ between accepting offer $P(v)$ and waiting one period for the next offer, which would be $\oP(t(v))$.  
Note that $P(\uv) = 2\uv^+$ because $t(\uv)=\uv$, and hence $P(v)\geq \max\{v,2\uv^+\}$ for all $v$. Consequently, $R(v)>0$ for all $v>\uv^+$.

The following result establishes that there is in fact an equilibrium corresponding to the pair of functions $(R,P)$. If $P$ is continuous, then on the equilibrium path Proposer first targets 
the threshold type $t(\ov)$ with offer $\oP(t(\ov))$, and then successively follows with offers $\oP(t^2(\ov)),\oP(t^3(\ov)),\ldots$. This is a decreasing sequence because $\oP$ and $t$ are increasing functions; the latter point owes to a monotone comparative statics argument. Vetoer accepts the initial offer if her type is in $[t(\ov),\ov]$, the second offer if her type is in $[t^2(\ov),t(\ov))$, the third offer if her type is in $[t^3(\ov),t^2(\ov))$, and so on.

\begin{lemma}\label{lemma:support_equilibrium}
Suppose $\uv\le 0$ or $\ov\le 1/2$. If $(R,P)$ supports a skimming equilibrium on $[\underline{v},\overline{v}]$ then there is an equilibrium in which proposals will be decreasing along the equilibrium path.
\end{lemma}

The proof of \autoref{lemma:support_equilibrium} builds on arguments from the seller-buyer bargaining literature \citep*[e.g.,][Theorem 1]{GSW:86}, and is relegated to the supplementary appendix. As discussed in the main text after \autoref{prop:existence_decreasing_offers}, novel considerations arise in deterring Proposer from deviating to offers below $2\uv^+$; for that we use \autoref{lemma:support_equilibrium}'s hypothesis that either $\uv\leq 0$ or $\ov\leq 1/2$. For readers familiar with the seller-buyer arguments, we also flag that another notable aspect of our argument is the use of the increasing envelope $\oP$. We use this because, owing to single-peaked Vetoer preferences, we cannot guarantee that there is a solution to equations \eqref{eq:Bellman} and \eqref{eq:Bellman2} in which the $P$ function is (even weakly) increasing. The lack of monotonicity precludes specifying $P(y)$ as type $y$'s acceptance threshold---we would not be assured that Proposer's beliefs are right-truncations. Using the increasing envelope $\oP$ to specify strategies allows us to surmount non-monotonicities in $P$.

For \autoref{lemma:support_equilibrium} to be useful, we must assure existence:

\begin{lemma}\label{lemma:DP_soln}
There is $(R,P)$ that supports a skimming equilibrium on $[\uv,\ov]$.
\end{lemma}
The proof of this result adapts arguments from the seller-buyer literature, and is relegated to the supplementary appendix. In a nutshell, we first suppose $\uv > 0$ and follow the reasoning of \citet*[pp.~78--79]{FLT:85} to show that there is an $(R,P)$ that supports a skimming equilibrium on $[\uv,\uv+\epsilon]$ if $\epsilon>0$ is small enough; the intuition is that when Proposer's belief is concentrated near $\uv$, the cost of delay outweighs the benefit from screening types and it is optimal to just offer $\oP(t(v))=2\underline{v}$ for all remaining types. An argument following \citet[Lemma A.3]{AD:89} allows us to extend $(R,P)$ to support a skimming equilibrium on $[\uv,\ov]$, proving \autoref{lemma:DP_soln} so long as $\uv>0$. Lastly, an approximation argument analogous to that in \citet[Theorem 4.2]{AD:89} allows us to cover the case of $\uv=0$, which in turn can be straightforwardly extended to $\uv<0$.

\begin{proof}[Proof of \autoref{prop:existence_decreasing_offers}]
Together, \Cref{lemma:support_equilibrium} and \Cref{lemma:DP_soln} establish a skimming equilibrium if either $\uv\le 0$ or $\ov\le 1/2$.

Let us show that Proposer's payoff in this equilibrium converges to ${\underline U(F)}$. Since Proposer never makes a strictly negative offer in this equilibrium and no type $v<0$ accepts a strictly positive offer, we assume without loss of generality that $\underline{v} \in [0,1/2)$. Let $A^*(v)$ denote type $v$'s choice from the menu $[2\underline{v},1]$.
As noted after \autoref{def:support_eq}, it holds that $P(v)\ge \max\{v,2\uv\}$. Hence, $P(v)\ge A^*(v)$. 

To show that Proposer's payoff is at least ${\underline U(F)}$ in the patient limit, observe that for any $v$ and any strictly positive integer $m$ there is $\underline{\delta}(m)$ such that for all $\delta>\underline{\delta}(m)$,
\begin{align}\label{eq:lower_bound_R}
    R(v) \ge (1-1/m) \int_{\underline{v}}^{v} \left[u(\min\{\oP(v'),1\}) -1/m\right] \mathrm dF(v'). 
\end{align} 
The intuition for this inequality is that if Proposer makes offers with small step size, he can ensure that each type $v$ accepts a proposal close to $\min\{\oP(v),1\}$, because each type $v$ accepts a proposal if and only if it is less than $\oP(v)$; moreover, as $\delta \to 1$ the cost of delay vanishes. Together with $\oP(v)\ge P(v)\ge A^*(v)$, inequality \eqref{eq:lower_bound_R} implies that if $(R,P)$ supports a skimming equilibrium then Proposer's payoff in this equilibrium is at least ${\underline U(F)}$ in the patient limit.

It remains to show that Proposer's payoff in any such equilibrium is at most ${\underline U(F)}$ in the patient limit.
Suppose not. Then there is $\epsilon'>0$ and a sequence $\delta_n\rightarrow 1$ such that for each $n$ there is $(R_n,P_n)$ supporting a skimming equilibrium that yields payoff at least ${\underline U(F)}+\epsilon'$. Let $A_n(v)$ be the proposal that is accepted in this equilibrium by type $v$ and let $\tau_n(v)$ be the time at which type $v$ accepts.\footnote{If type $v$ never accepts any proposal, we set $A_n(v):=0$ and $\tau_n(v):=\infty$.} Since $A_n$ is monotonic and uniformly bounded (as $0\le A_n(v)\le 1$ for all $v$ and $n$), Helly's selection theorem implies that there is a subsequence, which we also index by $n$ for convenience, along which $A_n\rightarrow A$ pointwise.

We claim $A(v)\ge v$ for all $v\le 1$.
Suppose not. Then there is $v\le 1$ 
and $\epsilon>0$ such that for all $n$ large enough, $A_n(v)\le v-\epsilon$. Let $x_n$ denote the state (in the sense described after \autoref{def:support_eq}) in which Proposer makes offer $A_n(v)$. Since $\oP_n(v)\ge v$, Proposer could offer $A_n(v)+\epsilon/2$ in state $x_n$ and get it accepted by all types in $[v-\epsilon/2, v]$, which have probability at least $\min\{\epsilon/2,v-\underline{v}\} \underline{f}$. For $\delta$ high enough such an offer is profitable, contradicting that $(R_n,P_n)$ supports a skimming equilibrium.

Since Proposer's payoff is at least ${\underline U(F)}+\epsilon'$, there must exist $v\in[2\underline v, \min\{\ov,1\}]$ 
and $\epsilon>0$ such that $A(v)= v+ \epsilon$ (by the dominated convergence theorem). Choose $v_1$ such that $A(v_1)=v+\epsilon$ and such that there is $v_2\ge v_1-\epsilon/5$ with $A(v_2)<A(v_1)$. We can then choose $\omega\in (0,\epsilon)$ such that $ A(v_2)\le v_1+\epsilon -\omega$. Since $v_1-\epsilon/5\le v_2 \le A(v_2)$, we can find $N$ such that for all $n>N$, $A_n(v_1)>v_1+ \epsilon - \omega/2$ and
\begin{align}\label{eq:bound_Av2}
v_1-\epsilon/4 \le A_n(v_2)\le v_1+\epsilon -3\omega/4. 
\end{align}

Let $s_n$ be the state in which Proposer makes offer $A_n(v_1)$ in equilibrium $(R_n, P_n)$. By definition, type $v_1$ accepts the offer $A_n(v_1)$ at time $\tau_n(v_1)<\infty$ (since $A_n(v_1)>0$) and therefore prefers $A_n(v_1)$ at time $\tau_n(v_1)$ over $A_n(v_2)$ at time $\tau_n(v_2)$. Moreover, the inequalities in \eqref{eq:bound_Av2} imply that type $v_1$ prefers $A_n(v_2)$ over $v_1+\epsilon -3\omega/4$. Hence,
$$
u_V(v_1+\epsilon- \omega/2, v_1)\ge \delta_n^{\tau_n(v_2)-\tau_n(v_1)} u_V(v_1+\epsilon -3\omega /4,v_1),$$
which rearranges to yield
$$\delta^{\tau_n(v_2)-\tau_n(v_1)} \le \frac{u_V(v_1+\epsilon-\omega/2, v_1)}{u_V(v_1+\epsilon -3\omega/4,v_1)}<1.$$
But this implies the following bound on $R_n$ in state $t_n(s_n)$ (after proposal $A_n(v_1)$ in state $s_n$ has been rejected; if the state is $\lim_{d'\uparrow s_n} t_n(d')$ the argument is analogous):
\begin{align}
 R_n(t_n(s_n)) \le & \int_{v_2}^{t_n(s_n)} u(\min\{\oP_n(v),1\}) \mathrm dF(v) \nonumber \\ 
 & +  \delta_n^{\tau_n(v_2)-\tau_n(v_1)} \int_{\underline{v}}^{v_2}u(\min\{\oP_n(v),1\}) \mathrm dF(v).   \label{eq:R_upperbound}
\end{align}
 To understand inequality \eqref{eq:R_upperbound}, note that for types above $v_2$ an upper bound on Proposer's utility is getting $\min\{\oP_n(v),1\}$ accepted immediately. Since type $v_2$, and therefore all lower types, cannot accept before waiting $\tau_n(v_2)-\tau_n(v_1)$ periods, an upper bound on Proposer's utility is getting $\min\{\oP_n(v),1\}$ accepted after $\tau_n(v_2)-\tau_n(v_1)$ periods.

For any strictly positive integer $m$, inequality \eqref{eq:lower_bound_R} implies that for all integers $n$ large enough, $$R_n(t_n(s_n))\ge (1-1/m)\int_{\underline{v}}^{t_n(s_n)} u(\min\{\oP_n(v),1\}) \mathrm dF(v)-1/m.$$ It follows that there exist $m$ and $n$ such that inequality \eqref{eq:R_upperbound} contradicts \eqref{eq:lower_bound_R}. 
\end{proof}

\subsection{A Commitment-Payoff Equilibrium}

\begin{lemma}
\label{lem:commitment_skimming}
Suppose $c^*>0$, and that either (i) $\uv < 0$ and $\Supp G=[\uv, 0]\cup [c^*/2,\ov]$ or (ii) $\uv=0$ and $\Supp G=[c^*/2,\ov]$.\footnote{We assume that $G$ has a density bounded away from $0$ and $\infty$ on its support.} 
There is a skimming equilibrium in which, on the equilibrium path, there is a decreasing sequence of proposals culminating in $c^*$, with Proposer payoff approximately $\int_{c^*/2}^{\ov} u(\max\{v,c^*\}) \mathrm{d}G(v)$.
\end{lemma}

\begin{proof}
First, by an argument analogous to \autoref{lemma:DP_soln}, there is $(R,P)$ that supports a skimming equilibrium on $[c^*/2,\overline{v}]$. Second, analogous to \autoref{lemma:support_equilibrium}, we can use that $(R,P)$ to construct a skimming equilibrium with the desired properties: 
just treat $c^*/2$ here like $\uv$ in \autoref{lemma:support_equilibrium}; the only point to note is that because in fact $\uv \leq 0$, no matter whether hypothesis (i) or (ii) in the lemma holds, we can deter deviations to any offer in $[0,c^*)$ by stipulating that any such offer is accepted, with the belief upon rejection supported on non-positive types and all subsequent offers being $0$.
\end{proof}

\begin{proof}[Proof of \autoref{prop:commitment_construction}]
We consider two cases, explaining in each case the beliefs and behavior off path that support the on-path behavior described in the proposition.

First, consider $\uv\leq 0$. If the first-period offer of $0$ is rejected, Bayes rule implies that Proposer updates to the belief $F_{[\uv, 0]\cup [c^*/2,\ov]}$, which is the prior's conditional distribution when excluding types $(0,c^*/2)$. Continuation play then follows the skimming equilibrium of \autoref{lem:commitment_skimming}. 
If Proposer makes a first-period offer other than $0$, continuation play follows the skimming equilibrium of \autoref{lemma:support_equilibrium}. It remains only to show that Proposer has no profitable deviation in the first period. Since Proposer's belief when his initial offer of $0$ is rejected is $F_{[\uv,0]\cup[c^*/2,\ov]}$, it follows from \autoref{lem:commitment_skimming} that Proposer's on-path payoff is approximately $\int_{c^*/2}^{1} u(\max\{v,c^*\}) \mathrm dF(v) + \int_{1}^{\ov} u(1) \mathrm dF(v)$, 
which equals $U(F)$. On the other hand, \autoref{prop:existence_decreasing_offers} implies that deviating to a first-period offer other than $0$ yields a payoff no more than approximately $\underline{U}(F)$. As $\underline{U}(F)< U(F)$, no such deviation is profitable.

Second, consider $\uv>0$ (and correspondingly $\ov\leq 1/2$). If the first-period offer of $0$ is rejected, Bayes rule implies that Proposer updates to the belief $F_{[c^*/2,\ov]}$. Continuation play then follows the skimming equilibrium of \autoref{lemma:support_equilibrium} applied to this belief, i.e., replacing $F$ in that lemma with $F_{[c^*/2,\ov]}$. If Proposer makes a first-period offer other than $0$, continuation play follows the skimming equilibrium of \autoref{lemma:support_equilibrium} with the original belief $F$. It follows from an essentially identical argument to that in the previous paragraph that no first-period deviation is profitable for Proposer.
\end{proof}

\subsection{Is Leapfrogging Necessary?}

\begin{proof}[Proof of \autoref{prop:leapfrogging_necessary}]
Towards contradiction, suppose there is a sequence of $\delta_n \to 1$ and corresponding skimming equilibria such that Proposer's payoff converges to $U(F)$. For each $n$ and $v$, let $B_n(v)$ denote the expected discounted proposal that type $v$ accepts: $B_n(v):= \E[\delta^t a_t]$, where the expectation is taken over the accepted proposals and agreement times for type $v$ given the equilibrium strategies. Since $B_n$ is monotonic (because the corresponding mechanism is IC) and uniformly bounded, Helly's selection theorem implies that there is some $B$ and a subsequence of $B_n$, which we also index by $n$ for convenience, along which $B_n\rightarrow B$ pointwise and in $L^1$-norm.

Since interval delegation is (essentially) uniquely optimal, it must hold that (up to measure zero sets) $B(v)=0$ for $v\in[\uv, c^*/2)$, $B(v)=c^*$ for $v\in(c^*/2,c^*)$, $B(v)=v$ for $v\in(c^*,\min\{\ov,1\})$, and $B(v)=1$ for $v\in[1,\ov]$. (Suppose not. $B$ corresponds to some feasible mechanism in the static problem and therefore, by the essential uniqueness assumption, yields payoff at most $U(F)-\varepsilon$ for some $\varepsilon>0$. Since $B_n\rightarrow B$ in the $L^1$-norm, for all $n$ large enough Proposer's payoff in the equilibrium corresponding $B_n$ is at most $U(F)-\varepsilon/2$, a contradiction.)

For any $\varepsilon>0$, there is $N$ such that for all $n>N$, $B_n(v)\le \varepsilon$ for all $v\in[\uv,c^*/2-\varepsilon]$. Then for all $n$ large enough, there is a history at which Proposer's belief is $F_{[\uv,c]}$ for some cutoff $c\ge c^*/2-\varepsilon$ (since on-path offers are accepted by upper sets) and Proposer's payoff in the continuation equilibrium is at most $u(\varepsilon)$. But, for any $\epsilon'\in(0,c)$, Proposer can deviate to make decreasing offers on a fine grid between $\epsilon'$ and $c$ such that all types in $[\epsilon',c]$ accept one of the offers close to their type or higher, and there is approximately no cost of delay as $\delta\rightarrow 1$.\footnote{One can verify that type $v\ge 0$ strictly prefers any action in $\left(v-v\sqrt{1-\delta},v+v\sqrt{1-\delta}\right)$ to action $v$ next period. Therefore, if Proposer makes decreasing offers between $\epsilon'$ and $c$ on a fine grid with diameter $\epsilon'\sqrt{1-\delta}$, every type in $[\epsilon',c]$ will accept one of the offers close to its type or higher in any best response. Moreover, agreement with those types is reached within at most $\frac{c-\epsilon'}{\epsilon' \sqrt{1-\delta}}+1$ rounds. Since $\delta^{\left(\frac{c-\epsilon'}{\epsilon' \sqrt{1-\delta}}+1\right)}\rightarrow 1$ as $\delta\rightarrow 1$, Proposer incurs essentially no cost of delay for types in $[\epsilon',c]$ as $\delta\rightarrow 1$.} Proposer's payoff from this deviation is strictly greater than $u(\epsilon)$ for $\epsilon$ and $\epsilon'$ small enough and $\delta$ large enough, contradicting Proposer's payoff in the continuation equilibrium being at most $u(\epsilon)$.
\end{proof}

\newpage
\section{Supplementary Appendix (For Online Publication)}

This supplementary appendix provides proofs for the lemmas stated in \appendixref{app:decreasing}. To reduce notation, we denote $S(v):= \oP(t(v))$. 

\begin{lemma}\label{lemma:maximizer_not_in_decreasing}
For any $v$ and $z<y \in T(v)$, we have $\oP(z)<\oP(y)$.
\end{lemma}

\begin{proof}
Suppose that there are $v$ and $z<y$ such that $\oP(z)\ge \oP(y)$. We prove that $y\notin T(v)$. Since $\oP$ is increasing, it is constant on $[z,y]$; call that value $\bar p$.
It follows that
\begin{align*}
& u(\bar p)[F(v)-F(y)]+\delta R(y)\\ \le & u(\bar p)[F(v)-F(y)]+\delta\left\{u(\bar p)[F(y)-F(z)]+ R(z)\right\}\\
 < & u(\bar p)[F(v)-F(z)]+\delta R(z),
\end{align*}
where the first inequality is because the payoff from any type in $[z,y]$ is at most $u(\bar p)$ (and hence $R(y)-R(z)\le u(\bar p)[F(y)-F(z)]$). Thus, $y\notin T(v)$.
\end{proof}

Below, we will use the fact that $T$ is upper hemicontinuous. This follows from the generalized theorem of the maximum in \citet[p.~527]{AD:89}. The theorem is applicable because: (i) the maximand function $ u(\oP(y))[F(v)-F(y)] +\delta R(y)$
 is upper semicontinuous as a function of $y$ for every $v$, which in turn is because $\oP$ is upper semicontinuous, and $u$ and $F$ are continuous and increasing on the relevant range $\{y: y\leq v \text{ and } \oP(y)\le 1\}$;\footnote{There is no loss in restricting attention to this range by a similar argument to that in the proof of \autoref{lemma:maximizer_not_in_decreasing}.} and (ii) for any sequence $v_n\rightarrow v$, the maximand function converges uniformly.

\begin{proof}[Proof of \autoref{lemma:support_equilibrium}]
\underline{Step 1}: We begin by specifying beliefs and strategies:
\begin{itemize}
    \item $\mu$ is derived from Bayes' rule whenever possible; if at history $h=(h',a)$ a probability $0$ rejection occurs, $\mu(h)$ puts probability $1$ on $\ov$ if $\ov\le 1/2$ and probability $1$ on $0$ if $\ov>1/2 $ (in the latter case, $\uv \leq 0$ by assumption);
    
    \item At any history $h=(h',a)$, any Vetoer type not in the support of Proposer's current belief plays an arbitrary best response; type $v \ge 0$ in the support accepts $a$ if and only if $a\in[0,\oP(v)]$; type $v<0$ in the support accepts if and only if $u_V(a,v)\ge u_V(0,v)$;
	
    \item Proposer's first offer is $S(\overline{v})$. To describe the rest of Proposer's strategy, consider any  history $h=(h',a)$.  Given Vetoer's strategy and the belief updating specified above, if Proposer holds a non-degenerate belief upon rejection of $a$ then this belief equals $F_{[\uv,d]}$ for some $d$. We stipulate that if 
    $a=\oP(d)=P(d)$, then Proposer offers $S(d)$;
    if $a=\oP(d)>P(d)$, then Proposer offers $\lim_{d'\uparrow d} S(d')$;
    if
    $a\in [\lim_{d'\uparrow d} \oP(d'), \oP(d))$, then Proposer randomizes between $\lim_{d'\uparrow d} S(d')$ and $S(d)$ so that type $d$ is indifferent between $a$ in the current period and the lottery in the next period; and for any $a\not\in[\oP(\uv),\oP(\ov)]$, Proposer offers $S(d)$. Finally, whenever Proposer's belief is degenerate on $x\ge 0$ ($x\in \{0,\ov\}$), Proposer offers $\min\{2x,1\}$ in all future periods. 
\end{itemize}

Observe that at any history, Proposer's subsequent on path offers are decreasing, either trivially if the current belief is degenerate, or for any non-degenerate belief because the belief cutoffs are decreasing by definition and $\oP$ and $t$ are increasing.

\underline{Step 2}: We verify that Proposer is playing a best response to Vetoer's strategy given beliefs $\mu$. As this is obvious whenever he has a degenerate belief, assume he has a non-degenerate belief.  As noted above, any such belief is of the form $F_{[\uv,d]}$ for some $d$. 
Proposer's strategy prescribes some randomization (possibly degenerate) between $S(d)$ and $\lim_{d'\uparrow d} S(d')$.

We first claim that $S(d)$ is an optimal proposal. Given Vetoer's strategy, $R(d)$ is an upper bound on Proposer's payoff. Furthermore, it follows from \autoref{lemma:maximizer_not_in_decreasing} that Vetoer's strategy has all types above $t(d)$ accepting $S(d)$ and all types strictly below rejecting. The claim follows.

We next claim that $\lim_{d'\uparrow d} S(d')$ is also an optimal proposal. Since $T$ is upper hemicontinuous, $\lim_{d'\uparrow d} t(d')\in T(d)$. Hence, given Vetoer's strategy, $\oP(\lim_{d'\uparrow d} t(d'))$ is an optimal proposal. It therefore suffices to show that $\lim_{d'\uparrow d} S(d')=\oP(\lim_{d'\uparrow d} t(d'))$, or equivalently, $\lim_{d'\uparrow d} \oP(t(d'))= \oP(\lim_{d'\uparrow d} t(d'))$.
Note that $\lim_{d'\uparrow d} \oP(t(d'))\le \oP(\lim_{d'\uparrow d} t(d'))$ because $t$ and $\oP$ are increasing. But if $\lim_{d'\uparrow d} \oP(t(d'))<\oP(\lim_{d'\uparrow d} t(d'))$ then continuity of $R$ and $u$ and strict monotonicity of $u$ in the relevant range imply the contradiction
\begin{align*}
R(d)&=u(\lim_{d'\uparrow d}\oP(t(d')))[F(d)-F(\lim_{d'\uparrow d}t(d'))]+ \delta R(\lim_{d'\uparrow d} t(d')) \\
&< u(\oP(\lim_{d'\uparrow d} t(d')))[F(d)-F(\lim_{d'\uparrow d}t(d'))]+ \delta R(\lim_{d'\uparrow d} t(d')) = R(d).
\end{align*}

All that remains is to verify that at a history $h=(h',a)$ with $a \in [\lim_{d'\uparrow d} \oP(d'),\oP(d))$, there is a randomization between $S(d)$ and $\lim_{d'\uparrow d} S(d')$ that makes type $d$ indifferent between $a$ in the current period and the lottery in the next period. To confirm this, note that since $P$ is right-continuous and $P(v)\ge v$ for any $v$, we have
\[ u_V(\lim_{d'\uparrow d}P(d'),d)\ge u_V(a,d)\ge u_V(P(d),d). \]
The existence of a suitable randomization now follows from continuity of $u_V(\cdot,d)$ and \Cref{eq:Bellman2}.

\underline{Step 3}: We verify that Vetoer is playing a best response at each history.
Consider any history $(h,a)$ with $\mu(h)=F_{[\uv,q]}$. Since types outside of the support of Proposer's belief play a best response by assumption, we only consider types in $[\uv,q]$. 
\begin{itemize}
    \item If $a>\oP(q)$, Vetoer's strategy prescribes that no type below $q$ accepts, and Proposer will propose $S(q)$ next period. Since type $q$ is indifferent between $P(q)$ in the current period and $S(q)$ next period, and $S(q)\le P(q)\le \oP(q)<a$, type $q$ prefers $S(q)$ next period to $a$ in the current period. The same holds for all lower types, and hence Vetoer is playing a best response.
    \item If $a<0$, then: (i) it is clearly a best response for all types $v\ge 0$ to reject; and (ii) types $v<0$ accept if and only if they prefer $a$ to $0$, which is a best response because Proposer will never make a strictly negative offer in the continuation equilibrium.
    \item If $a$ is positive but below the range of $\oP$, all types $v\ge0$ accept. After a rejection, Proposer  
    will either perpetually offer $0$ or $2\ov$, yielding a continuation payoff of $0$ to all types, and so it is a best response for any type $v\geq 0$ to accept $a$.
    \item Otherwise, $a$ is between $\oP(\uv)$ and $\oP(q)$. 
    
    If $a=\oP(d)=P(d)$ for some $d\le q$, Vetoer's strategy prescribes that all and only those types above $d$ accept.\footnote{If there are multiple values of $d$ satisfying $a=\oP(d)$, all types above the lowest one accept.} On path, Proposer will propose $S(d)$ next period followed by lower offers; since type $d$ is indifferent between $a$ in the current period and $S(d)$ next period, and all future offers are below $a$, SCED implies that it is a best response for all higher types to accept and for all lower types to reject. Hence, Vetoer is playing a best response.
    
    If there is $d\le q$ such that $a=\oP(d)>P(d)$, Vetoer's strategy prescribes that all and only those types above $d$ accept. Proposer will propose $\lim_{d'\uparrow d} S(d')$ next period, followed by lower offers. Since type $d'$ is indifferent between $P(d')$ in the current period and $S(d')$ next period, continuity of $u$ implies that type $d$ is indifferent between $\lim_{d'\uparrow d} P(d') = \oP(d)=a$ in the current period and  $\lim_{d'\uparrow d} S(d') $ next period. Hence, Vetoer is playing a best response.
    
    If there is $d\le q$ such that $a\in[\lim_{d'\uparrow d} \oP(d'),\oP(d))$, Vetoer's strategy again prescribes that all and only those types above $d$ accept. Proposer will randomize next period between $\lim_{d'\uparrow d}S(d')$ and $S(d)$ to make type $d$ indifferent between accepting $a$ or getting the lottery next period. Therefore, Vetoer is playing a best response. 
    \qedhere
\end{itemize}
\end{proof}

 \begin{proof}[Proof of \autoref{lemma:DP_soln}]
 
 \underline{Step 1}: Suppose $\uv>0$.  We claim that there is $\epsilon>0$ such that $(R,P)$ given by 
\begin{align*}
    R(v)&:= u(2\uv) F(v)\\
    P(v)&:=v+\sqrt{v^2-4\delta \uv(v-\uv)}
\end{align*}
supports a skimming equilibrium on $[\underline{v},\underline{v}+\epsilon]$. Plainly, $R$ and $P$ are continuous, given that 
$F$ is continuous. Also, $P$ is increasing and hence $\oP=P$. Some algebra confirms that $R(v)$ is the value from securing acceptance from all types below $v$ on action $2 \uv$, while $P(v)$ is the action that makes type $v$ indifferent between accepting that action now and getting action $2\uv$ in the next period. Therefore, it is sufficient for us to show that there is $\epsilon>0$ such that for all $v\in[\uv,\uv+\epsilon]$ the unique maximizer of the RHS of \Cref{eq:Bellman} is $\uv$, which implies $t(v)=\uv$.

To that end, observe that the derivative of the objective function in \Cref{eq:Bellman} with respect to $y$ is
 \begin{align}\label{eq:derivative_Bellman}
     u'(\oP(y)) \oP'(y) [F(v)-F(y)]- u(\oP(y)) f(y) + \delta u(2\uv) f(y).
 \end{align}
Since $0<u(2\uv)\le u(\oP(y))$ and $f$ is bounded away from 0, the sum of the last two terms in expression \eqref{eq:derivative_Bellman} is strictly negative and bounded away from 0. Since $u'(\oP(y))$ is bounded (by concavity), $\oP'(y)$ is bounded (as $v^2- 4 \delta \underline{v}(v-\underline{v})>0$ for all $v$), $F$ is continuous, and $v,y\in[\uv,\uv+\epsilon]$, the first term in expression \eqref{eq:derivative_Bellman} goes to $0$ as $\epsilon \to 0$. It follows that there is $\epsilon>0$ such that expression \eqref{eq:derivative_Bellman} is strictly negative for all $y\in [\uv,\uv+\epsilon]$, and hence the maximum of the RHS of \Cref{eq:Bellman} is attained uniquely at $t(v)=\uv$ whenever $v\le \uv+\epsilon$. 

\underline{Step 2}: Suppose $(R_{v^*},P_{v^*})$ supports a skimming equilibrium on $[\uv,v^*]$, where $0<\uv<v^*<\ov$. We will show that there is $(R,P)$ that supports a skimming equilibrium on $[\uv,\ov]$ with the property that $P(v)=P_{v^*}(v)$ and $R(v)=R_{v^*}(v)$ for all $v\in[\uv,v^*]$.

Pick $v'\in (v^*, \overline{v}]$ as large as possible such that 
 \begin{align}
      u(1)[F(v')-F(v^*)]\le (1/2) (1-\delta)R_{v^*}(v^*).\label{eq:choice_of_v'}
 \end{align}
Note that $v'$ is well-defined because $F$ is continuous and $R_{v^*}(v^*)>0$ (this inequality holds because of $v^*>\uv$ and the property noted at the end of the paragraph following \autoref{def:support_eq}). Moreover, letting $\overline f$ denote an upper bound for $f$, it holds that
\begin{align}\label{eq:extension_bound}
v'-v^* \ge \frac{(1/2)(1-\delta)R_{v^*}(v^*)}{u(1)\overline{f}}>0.    
\end{align}

We extend $R_{v^*}$ to $R_{v'}$ defined on $[\uv,v']$ by setting $R_{v'}(v):= R_{v^*}(v)$ for $v\in [\uv,v^*]$, and for $v\in (v^*,v']$,
 \begin{align*}
     R_{v'}(v):=\max_{y\in[\uv,v^*]} \left\{u(\oP_{v^*}(y))[F(v)-F(y)] +\delta R_{v^*}(y)\right\}
 \end{align*}
 and define $t_{v'}(v)$ to be the largest value in the argmax correspondence. Observe that $\oP_{v^*}$ is upper semicontinuous (since $P_{v^*}$ is right-continuous by assumption, and hence $\oP_{v^*}$ is right-continuous) and $R_{v^*}$ is continuous; hence, $R_{v'}(v)$ and $t_{v'}(v)$ are well-defined. We extend $P_{v^*}$ to $P_{v'}$ defined on $[\uv,v']$ by setting $P_{v'}(v):=P_{v^*}(v)$ for $v\in[\uv,v^*]$, and for $v\in(v^*,v']$ by letting $P_{v'}(v)$ be the largest value satisfying
 \begin{align*}
     u_V(P_{v'}(v),v)=\delta u_V(\oP_{v^*}(t_{v'}(v)),v).
 \end{align*}
 So $(R_{v'},P_{v'})$ satisfies \Cref{eq:Bellman2}.  We can apply the  generalized theorem of the maximum in \citet[p.~527]{AD:89} analogously to the discussion after \autoref{lemma:maximizer_not_in_decreasing} and conclude that $R_{v'}$ is continuous and $T_{v'}$ is non-empty and upper hemicontinuous. Therefore, $t_{v'}$ is upper semicontinuous and, since it is increasing, right-continuous. 
 These properties of $t_{v'}$ and the hypothesis that $P_{v^*}$ is right-continuous imply that $P_{v'}$ is right-continuous.
 $(R_{v'},P_{v'})$ also satisfies \Cref{eq:Bellman}, i.e.,
 \begin{align*}
     R_{v'}(v)=\max_{y\in[\uv,v]} \left\{u(\oP_{v'}(y))[F(v)-F(y)] +\delta R_{v'}(y)\right\}
 \end{align*}
 for all $v\in[\uv,v']$, because for all $y\in[v^*,v]$, 
 \begin{align*}
     &u(\oP_{v'}(y))[F(v)-F(y)] +\delta R_{v'}(y)\nonumber\\
     \le &u(1)[F(v)-F(y)] +\delta R_{v'}(y)\nonumber\\
     \le &(1/2) (1-\delta)R_{v^*}(v^*)+\delta R_{v'}(y)\nonumber \\
     \le &(1/2) (1-\delta)R_{v'}(y)+\delta R_{v'}(y)\nonumber \\
     < &R_{v'}(v).
 \end{align*}
Here the second inequality is because the choice of $v'$ satisfies inequality \eqref{eq:choice_of_v'} and the second inequality is because $R_{v^*}(v^*)=R_{v'}(v^*)$ and $R_{v'}$ is increasing. Therefore, the maximum is attained for $y\in [\underline{v},v^*)$ and the claim follows since $R_{v'}(y)=R_{v^*}(y)$ for any such $y$.

We have established that $(R_{v'},P_{v'})$ supports a skimming equilibrium on $[\uv,v']$. Since $R_{v'}$ is increasing, it follows from inequality \eqref{eq:extension_bound} that a finite number of repetitions of this argument extends $(R_{v^*},P_{v^*})$ to the entire $[\uv,\ov]$ interval.
 
\underline{Step 3}: By an approximation argument analogous to that in \citet[Theorem 4.2]{AD:89}, there exists $(R,P)$ that supports a skimming equilibrium on $[\uv,\ov]$ if $\uv=0$; we omit details.  The case of $\uv<0$ is handled by setting $R(v)=0$ and $P(v)=0$ for all $v<0$, and pasting that to a solution when we take $\uv=0$ and set the distribution on $[0,\ov]$ to be the conditional distribution $F_{[0,\ov]}$.
 \end{proof}
 
 \end{document}